\newif\ifconfver
\confvertrue        

\ifconfver
\documentclass[10pt,twocolumn,twoside]{IEEEtran}
\else
\documentclass[11pt,draftcls,onecolumn]{IEEEtran}
\fi
\IEEEoverridecommandlockouts
\usepackage{amsfonts,amsmath,amsfonts, amssymb, amstext, graphicx,bm,amsthm,color}
\usepackage{cite}
\usepackage{epstopdf}
\usepackage{anyfontsize}
\usepackage{url,verbatim}
\usepackage{enumerate} 
\usepackage{hyperref}
\usepackage{latexsym}
\usepackage{mathrsfs}
\usepackage{academicons}
\usepackage{graphicx,subfigure}
\usepackage{epsfig,multirow}
\usepackage{algorithmic,algorithm}
\usepackage{comment}
\usepackage{verbatim}
\usepackage{diagbox}

\newtheorem{lemma}{Lemma}
\newtheorem{thm}{Theorem}

\newcommand{\limto}{\rightarrow}

\newcommand{\bz}{\mathbf 0}

\newcommand{\R}{\mathbb R}

\newcommand{\C}{\mathbb C}
\newcommand{\E}[1]{{\mathbb E}\left[ #1 \right]}
\newcommand{\tabincell}[2]{\begin{tabular}{@{}#1@{}}#2\end{tabular}}

\begin{document}
\title{Quartic Perturbation-based Outage-constrained Robust Design in Two-hop One-way \\Relay Networks}
\author{Sissi Xiaoxiao Wu,  Sherry Xue-Ying Ni, Jiaying Li, and Anthony Man-Cho So

\thanks{This work is supported by the National Natural Science Foundation of China under Grant 61701315; by Shenzhen Technology R$\&$D Fund JCYJ20170817101149906 and JCYJ20190808120415286; by Shenzhen University Launch Fund 2018018. S. X. Wu and J. Li are with the College of Electronics and Information Engineering, Shenzhen University, Shenzhen, China. S. X.-Y. Ni is with NM OPTIM Limited. A. M.-C. So is with the Department of Systems Engineering and Engineering Management, The Chinese University of Hong Kong, Shatin, N.T., Hong Kong S.A.R., China.  
	E-mails: 
	{\{xxwu.eesissi\}@szu.edu.cn, manchoso@se.cuhk.edu.hk, sherry.ni@nm.dev}.
}

}

\date{\today}
\maketitle

\begin{abstract}
In this work, we study a classic robust design problem in two-hop one-way relay system. We are particularly interested in the scenario where channel uncertainty exists in both the transmitter-to-relay and relay-to-receiver links. By considering the problem design that minimizes the average amplify-and-forward power budget at the relay side while satisfying SNR outage requirements, an outage-constrained robust design problem involving quartic perturbations is formulated to guarantee the robustness during transmission. This problem is in general difficult as it involves constraints on the tail probability of a high-order polynomial. Herein, we resort to moment inequality and Bernstein-type inequality to tackle this problem, which provide convex restrictions, or safe approximations, of the original design. We also analyze the relative tightness of the two safe approximations for a quadratic perturbation-based outage constrained problem. Our analysis shows that the Bernstein-type inequality approach is less conservative than the moment inequality approach when the outage rate is within some prescribed regime. To our best knowledge, this is the first provable tightness result for these two safe approximations. Our numerical simulations verify the superiority of the robust design and corroborate the tightness results.  
\end{abstract}

\section{Introduction}
In recent decades, multi-hop relay technology has been widely used in long-distance wireless communication systems to expand the coverage of communication. For example, when devices are far away from each other, such a technology can improve the “quality-of-service” (QoS) between transmitters and receivers \cite{ieee2009}. Moreover, multi-hop relay technology can be applied to other advanced communication systems, such as device-to-device (D2D) communication for which the equipment can act as an instant relay node~\cite{zhang2015intra,tian2019analysis, gong2018backscatter}, the millimeter wave communication which can overcome the serious attenuation phenomenon at high frequency~\cite{roh2014millimeter}, and cognitive radio (CR) networks in which it could improve the coverage of cognitive network and the channel capacity of the system~\cite{cai2016energy,zhang2017optimized}. In the multi-hop relay transmission, a classic task is to design the amplify-and-forward (AF) weights that guide the relays to adjust their antennas towards the receiver. This design task is conditioned on that the relay system acquires the channel state information (CSI) from both the transmitter and the receiver. However, in practice, due to estimation error, quantization error, or limited feedback, the available CSIs at the relays are usually imperfect. It is well known that CSI uncertainties might lead to serious performance degradation during transmission. This motivates us to study robust designs for the multi-hop relay system.

In this paper, we consider the robust design problem in the context of two-hop one-way relay beamforming, which is generally more involved than its non-robust counterpart \cite{fazeli2009multiple}. This robust design problem is considered in either a worst-case setting or a chance-constrained setting. In both settings, it is usually assumed that CSI errors are only present in either the transmitter-to-relay link or the relay-to-receiver link.
Under this assumption, the so-called S-lemma can be applied to the worst-case setting to turn a semi-infinite program (SIP) \cite{stein2012solve, lasserre2015tractable} into a tractable semidefinite program (SDP) \cite{zheng2009robust, RobustrelaybeamformingCui12}; while in the chance-constrained setting, the so-called sphere-bounding or Bernstein-type inequality \cite{wang2014outage,ni2019outage,ni2018outage,zhou2020robust,yazar2020power,keskin2018optimal,li2018robust} can be applied to find safe approximations of the original chance constraints. 
So far, there are very few works considering a reliable robust design for the case where both the transmitter-to-relay and relay-to-receiver links have errors, as this case involves quartically perturbed constraints and is generally difficult. Our previous work \cite{wu2016polynomial} solved a quartically-constrained robust design problem in a worst-case setting. Our goal in this paper is to fill the gap in the chance-constrained setting. In this work, we target the problem design that minimizes the average AF power at the relays while satisfying the receivers’ signal-to-noise ratio (SNR) constraints with a small outage probability, given that both the transmitter-to-relay  and relay-to-receiver links are subject to Gaussian errors. To tackle the resulting outage constraints with quartic polynomials of complex Gaussian random variables \cite{ma2014robust}, there are two possible approaches: 1) Reformulate the high-order chance constraints into safe tractable approximations and employ the positive semidefinite relaxation (SDR) technique \cite{luo2010semidefinite} to solve the approximated problem;  2) simply ignore the higher-order perturbation terms and deal only with quadratic chance constraints. For the former, we can resort to a suitable moment inequality to develop safe approximations of the quartically perturbed outage constraints. For the latter, we can apply both the moment inequality and Bernstein-type inequality to develop safe approximations of the approximating quadratically perturbed outage constraints. The resulting approximations can then be 
solved by the SDR technique, and a sub-optimal AF beamforming vector can be extracted from the optimal SDR solution using a Gaussian randomization procedure.

In the literature, the common robust design only gives rise to a quadratically perturbed chance constraint. For example, \cite{wang2014outage} first proposed and summarized three methods to solve chance constraints involving quadratic forms, namely sphere bounding, Bernstein-type inequality, and decomposition-based large deviation inequality (LDI). The works \cite{yan2019outage,zhang2019robust} applied Taylor's expansion to approximate the chance constraint of interest by one with quadratic perturbations and then tackled it using the LDI approximation. In \cite{ni2019outage,ni2018outage,zhou2020robust}, the authors used Bernstein-type inequality to convert a quadratically perturbed chance constraint into a deterministic form in the robust design of CR networks. The sphere bounding method was used to propose a safe tractable approximation of the quadratic chance constraint in \cite{yazar2020power,keskin2018optimal}. Furthermore, \cite{li2018robust} used the S-procedure and Bernstein-type inequality to tackle quadratic chance constraints and compared the performance of these two methods. Even for works that originally aim at tackling robust designs with high-order perturbations, the traditional way is to ignore the higher-order terms and apply standard techniques from robust optimization to simplify the constraints~\cite{chalise2009mimo,tao2012robust,aziz2012robust,ponukumati2013robust}. There are also other works that tackle quartic constraint but not from a probabilistic perspective. For example,  \cite{wu2016polynomial}  employed the SDR technique and tools from polynomial optimization to construct safe approximations of such constrains; \cite{shi2019spectrally} introduced several auxiliary variables to convert the quartic ISL constraint and the PAPR constraint into several quadratic constraints and proposed an alternating direction method of multipliers-based solution to handle them;  \cite{jin2018hybrid} showed that the quartic constraint associated with a certain hybrid precoding problem is automatically satisfied and hence can be removed; \cite{wen2017joint} introduced a slack variable to replace the higher-order terms in the SINR QoS constraint, thus converting it into a linear matrix inequality. Generally speaking, finding a good solution to robust designs involving quartic perturbations is a very difficult problem, especially from a probabilistic perspective. 
Recently, \cite{ma2014robust} proposed a safe tractable approximation of quartically perturbed chance constraints using moment inequalities for Gaussian polynomials and the SDR technique. However, no theoretical analysis was provided, and the relative tightness of the Bernstein-type inequality approach and the moment inequality approach remains unknown.


Our contribution in this work is fourfold. First, we introduce the two-hop relay robust design problem with a quartically perturbed chance constraint, which is seldom studied in prior work. Second, we apply
the fourth-order moment inequality to obtain a safe approximation of the said chance constraint. Third, we apply the second-order moment inequality and the Bernstein-type inequality to provide a safe approximation of the quadratically perturbed chance constraint obtained by ignoring higher-order perturbation terms. In addition, we provide an analytical bound to prove the relative tightness of different restrictions. Our numerical results show that the proposed approximation approaches are more reliable than the non-robust counterpart. Also, the results suggest that
by dealing with all the perturbation terms in the chance constraint instead of keeping only the lower-order ones, the resulting design is more robust against perturbations that do not match the prior distributional information. Lastly, our comparison between the second-order moment inequality-based approach and the Bernstein-type inequality-based approach corroborate our relative tightness result.

The rest of the paper is organized as follows. In Section \ref{sec:2}, we provide the system model and formulate the robust design problem for the two-hop one-way relay system.
Section \ref{sec:3} introduces the moment inequality-based approach, which can provide safe approximations for both the quartically perturbed chance constraint and the quadratically perturbed chance constraint. In Section \ref{sec:4}, the Bernstein-type inequality-based approach is proposed to find a safe approximation of the quadratically perturbed chance constraint. Moreover, in Section \ref{sec:5}, we establish theoretically the relative tightness of the moment inequality-based and Bernstein-type inequality-based approaches.
Simulation results are presented in Section \ref{sec:sim} and
we conclude our work  in Section \ref{sec:con}.

\section{System Model and Problem Formulation} \label{sec:2}
In this work, we consider a classic scenario setting for a relay network consisting of one transmit-receiver pair and there are $L$ relays between them to assist the transmission. We assume that no direct link is  involved in this setting and both the transmitter and receiver are equipped with a single antenna. Then, the information is transmitted from the transmitter to receiver through two types of links.
One is the \emph{transmitter-to-relay links}, via which the transmitter sends common information to the relays. In this context,
the receive model is given by
\begin{equation}\label{rt}
{\bm r}(t) = {\bm f}s(t) + {\bm n}(t),
\end{equation}
where $s(t)$ is the common information with $\mathbb{E}[|s(t)|^2]=P_t$ and $P_t$ is the transmit power at the transmitter; ${\bm f} \in \mathbb{C}^{L}$ is the channel from the transmitter to the relays; ${\bm n}(t)=[n_1(t),\ldots,n_\ell(t),\ldots ,n_L(t)]^T$ and $n_\ell(t)$ is the white noise at relay-$\ell$ with variance $\sigma_\ell^2$.
The other is the \emph{relay-to-receiver links}, via which relays amplify and forward the received signal to the receiver. In this paper, we target at the relay beamforming scheme, in which
the AF process at the relay side is given by
\begin{equation} \label{eq:xt}
{\bm x}(t) = {\rm Diag}({\bm w}) {\bm r}(t),\footnote{The operator ${\rm Diag}({\bm v})$ will output a diagonal matrix with the elements of the vector ${\bm v}$ on the diagonal.}
\end{equation}
where 
${\bm w} = [w_1,\ldots,w_\ell, \ldots,w_L]^T$ and $w_\ell$ is the AF weight at relay $\ell$.
Under this model, the received signal can be expressed as
\begin{align}\label{eq:yt}
y(t) = &{\bm g}^H{\bm x}(t) + {v}(t),
\end{align}
where ${\bm g} \in \mathbb{C}^L$ is the channel from the relays to the receiver; ${v}(t)$ is the white noise at the receiver with variance $\sigma_{{v}}^2$. Then, the SNR at the receiver can be expressed as
\[
{\rm SNR} = \frac{{\bm w}^H P_t({\bm f}\odot {\bm g}^*)({\bm f}\odot {\bm g}^*)^H{\bm w}}{{\bm w}^H {\rm Diag} ([|g^1|^2\sigma_1^2, |g^2|^2\sigma_2^2,\ldots,|g^L|^2\sigma_L^2]){\bm w}+\sigma_{{v}}^2}.
\]

\begin{figure}[t!]
	\centering
	\includegraphics[width = 8cm]{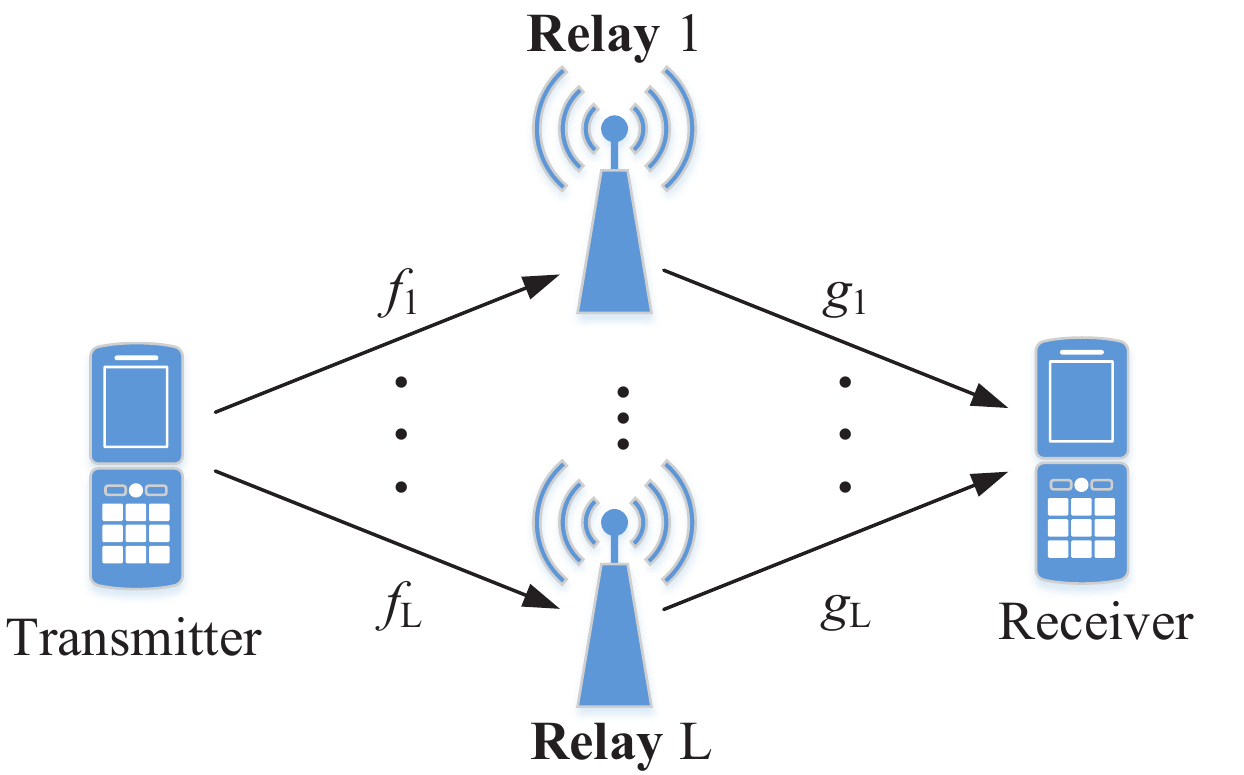}
	\caption{The two-hop one-way relay network. }
	\label{Fig:1}
\end{figure}

Under \eqref{eq:xt} and \eqref{eq:yt} we may explicitly express the power at the relays as ${\rm \mathbb{E}}[{\bm w}^H {\bm D}{\bm w}]$ and ${\rm SNR}$ as
\begin{equation*}\label{P_SNR}
\begin{array}{ll}
& \quad \displaystyle {\rm SNR}= \frac{{\bm w}^H {{\bm A}}{\bm w}}{{\bm w}^H {{\bm C}}{\bm w}+1},
\end{array}
\end{equation*}
where
\begin{align}\label{ACD}
{{\bm A}} = &P_t(({\bm f}{\bm f}^H)\odot (({\bm g}^*)({\bm g}^*)^H))/\sigma_{{v}}^2, \\ \label{ck}
{{\bm C}} = &{\bm \Sigma}\odot ({\bm g}{\bm g}^H)/\sigma_{{v}}^2, \quad 
 \\\label{d}
{\bm D} = &P_t{\bm I}\odot({\bm f}{\bm f}^H) + {\bm \Sigma},
\end{align}
$\sigma_{{v}}^2$ is the noise power at the destination,  ${\bm \Sigma}={\rm Diag}([\sigma_1^2,...,\sigma_L^2])$, and ${\bm f}$, ${\bm g}$ represent the actual transmiter-to-relays channel and relays-to-receiver channel, respectively.  In the general case where both links are imperfect, we write
\[
{\bm f} =  {\bar {\bm f}} + \Delta{\bm f}, \quad \quad {\bm g} = {\bar {\bm g}} + \Delta{\bm g},
\]
where ${\bar {\bm f}}$ and ${\bar {\bm g}}$ are the estimated CSI;  $\Delta{\bm f}$ and $\Delta{\bm g}$ are the corresponding stochastic CSI errors that respectively follow the distributions $\Delta{\bm f}\sim\mathcal{CN}(0,{{\bm E}}_f)$, ${{\bm E}}_f\succ 0$ and $\Delta{\bm g}\sim\mathcal{CN}(0,{{\bm E}}_g)$, ${{\bm E}}_f\succ 0$.  
Here we adopt a Gaussian channel error model; $i.e.$, ${{\bm E}}_f=\epsilon^2{\bm I}$, ${{\bm E}}_g=\eta^2{\bm I}$ with $\epsilon,\eta>0$. Equivalently, we may write
\[
\Delta{\bm f}=\epsilon{\bm x}, \quad \quad \Delta{\bm g}=\eta{\bm y},
\]
where ${\bm x}$ and ${\bm y}$ are standard complex Gaussian vectors and $\epsilon, \eta$  are known scalars to bound the error magnitudes~\cite{fazeli2009multiple, chalise2009mimo}.

We are motivated to design the AF weight vector ${\bm w}$ so that the average transmit power is minimized while the receiver's SNR outage constraint is satisfied. Specifically, we consider the following problem:
\begin{equation}\label{main0}
\begin{array}{ll}
\displaystyle\min_{{\bm w} \in \mathbb{C}^L} &  \quad \displaystyle {\rm \mathbb{E}}[{\bm w}^H {\bm D}{\bm w}]\\
\text{subject to} \quad &\quad \displaystyle {\rm Pr}\{{\rm SNR}\le \gamma\} \le \rho, \\
\end{array}
\end{equation}
where $\gamma$ is the target SNR threshold and $\rho$ is a prescribed outage rate we want to guarantee during the transmission. 
To further tackle Problem \eqref{main0}, a typical relaxation is

\begin{equation}\label{main1}
\begin{array}{ll}
\displaystyle\min_{{\bm W} \in \mathbb{C}^{L \times L}} &  \quad \displaystyle {\rm \mathbb{E}}[{\bm D} \cdot {\bm W}]\\
\text{subject to} \quad &\quad \displaystyle {\rm Pr}\{ Q({\bm W},{\bm x},{\bm y})\ge 0\} \le \rho, \\
\quad &\quad  {\bm W} \succeq 0,
\end{array}
\end{equation}
where 
\begin{align}\label{SNR_exact}
&Q({\bm W},{\bm x},{\bm y}) \\\notag
=&\sigma_{{v}}^2+{\bm W}\cdot \big( {\bm \Sigma}\odot ({\bm g}{\bm g}^H) - \frac{P_t}{\gamma}(({\bm f}{\bm f}^H)\odot (({\bm g}^*)({\bm g}^*)^H))\big).
\end{align}
Problem \eqref{main1} is still difficult as $Q({\bm W},{\bm x},{\bm y})$ in the outage constraint involves high-order perturbation terms. Our main task in the sequel is to discuss how to deal with this challenge.

\section{Moment Inequality-Based Approach} \label{sec:3}
In this section, we review and further develop the moment inequality approach in \cite{ma2014robust} to construct safe tractable approximations of chance constraints with quartic perturbations.

\subsection{The Fourth-order Moment Inequality Method}
Key to the development of safe tractable approximations of
quartically perturbed chance constraints is the following moment
inequality for quartic polynomials in complex Gaussian
random variables.
\begin{thm}
	 Let $\xi_1,\ldots,\xi_m$ be independent standard real Gaussian random variables.  Consider the function $f:\C^n \times \R^m \limto \R$:
\begin{align}
	&f({\bm x},{\bm \xi}) =-a_0({\bm x}) + \sum_{1\le i \le m} \xi_i a_i({\bm x}) + \sum_{1\le j_1,j_2 \le m} \xi_{j_1}\xi_{j_2} a_{j_1j_2}({\bm x}) \notag\\
	&+\sum_{1\le k_1,k_2,k_3 \le m} \xi_{k_1}\xi_{k_2}\xi_{k_3}a_{k_1k_2k_3}({\bm x}) \\\notag
	&+ \sum_{1\le \ell_1,\ell_2,\ell_3,\ell_4 \le m} \xi_{\ell_1}\xi_{\ell_2}\xi_{\ell_3}\xi_{\ell_4} a_{\ell_1\ell_2\ell_3\ell_4}({\bm x}),
\end{align}
where $a_0,a_i,a_{j_1j_2},a_{k_1 k_2 k_3},a_{\ell_1 \ell_2 \ell_3 \ell_4}$ are affine functions of $ \bm x$.  Note that we allow the decision vector ${\bm x}$ to take complex values.  However, we assume that the value $f({\bm x},{\bm \xi})$ is real for any ${\bm x} \in \C^n$ and ${\bm \xi} \in \R^m$.  Consider the chance constraint
\begin{equation}\label{chance_Q}
\Pr (f({\bm x},{\bm \xi})\ge 0)\le \rho,
\end{equation}
where $\rho > 0$ is given. The following hold:
\begin{itemize}
\item[(a)] For each ${\bm x} \in {{\mathbb C}^{n}}$ and ${\bm \xi} \in {{\mathbb R}^{m}}$, let  
\[
\bar{f}({\bm x},{\bm \xi})\triangleq f({\bm x},{\bm \xi})+{{a}_{0}}({\bm x}).
\] 
Then, the function $\bm{x} \mapsto \bar{f}( \bm{x}, \bm{\xi})^2$ is quadratic in $\bm{x}$ and can be written in the form
\begin{align}
\bar{f}({\bm x},{\bm \xi})^{2}={\bm v}^H({\bm x}){\bm U}({\bm \xi}){\bm v}({\bm x})
\end{align}
for some  ${\bm U}({\bm \xi}) \succeq  {\bm 0}$, where ${\bm U}({\bm \xi})$ is a Hermitian positive semidefinite matrix whose rows and columns are labeled by the set of indices
$$ \mathscr{S} = \{0,\underbrace{i,\ldots}_{1\le i \le m},\underbrace{j_1j_2,\ldots}_{1\le j_1,j_2 \le m},\underbrace{k_1k_2k_3,\ldots}_{1\le k_1,k_2,k_3 \le m},\underbrace{\ell_1\ell_2\ell_3\ell_4,\ldots}_{1\le \ell_1,\ell_2,\ell_3,\ell_4 \le m}\} $$
arranged in lexicographic order and ${\bm v}:\C^n \limto \C^{|\mathscr{S}|}$ is an affine function. In particular, ${\bm v}({\bm x})$ is a vector whose $s$-th component is $a_s({\bm x})$, where $s \in \mathscr{S}$.

\item[(b)] Let ${\bm U} \triangleq {\mathbb E}[{\bm U}({\bm {\xi}})]\succeq {\bm 0}$ and
\begin{equation}\label{c(epsilon)}
c(\rho)\triangleq \left\{
\begin{array}{ll}
{{({q}(\rho)-1)}^{2}}\exp (\frac{2{q}(\rho)}{{ q}(\rho)-1})   &\ {\rm if}\ {q}(\rho)>2;\\
1/{\sqrt{\rho}}\ &\ {\rm if}\ {q}(\rho)=2,
\end{array}
\right.
\end{equation}
where
\begin{equation}\label{q(epsilon)}
{\mathop{q}}\,(\rho)\triangleq \left\{
\begin{array}{ll}
\frac{-\ln \rho+\sqrt{{{(\ln \rho)}^{2}}-8\ln \rho}}{4}   &\ {\rm if}\ \rho \in (0,\exp (-8)];\\
2\ &\ {\rm otherwise},\
\end{array}
\right.
\end{equation}
Then, the second-order cone constraint
\begin{equation}\label{SOC constrain}
{{a}_{0}}({\bm x})\ge c(\rho)||{{\bm U}^{1/2}}{\bm x}||
\end{equation}
serves as a safe tractable approximation of the chance
constraint in \eqref{chance_Q}.
\end{itemize}
\label{thm:1}
\end{thm}	
This theorem was first proposed in \cite{ma2014robust}  for robust beamforming design in a two-way relay network. Therein, a partial proof was sketched while important details were omitted due to the page limit. In this work, we will give a complete proof in Appendix \ref{appendix:1}. Now, let us apply Theorem \ref{thm:1} to our problem. We define
\begin{align*}
&a_0({\bm W})\\
=&-\sigma_{{v,d}}^2 - {\bm W}\cdot \big( {\bm \Sigma}\odot (\bar{\bm g}\bar{\bm g}^H) - \frac{P_t}{\gamma}((\bar{\bm f}\bar{\bm f}^H)\odot (({\bar{\bm g}}^*)({\bar{\bm g}}^*)^H))\big).
\end{align*}
Then, we have
\begin{eqnarray*}
\bar{f}({\bm W},{\bm x},{\bm y})&=& f({\bm W},{\bm x},{\bm y})+a_0({\bm W})\\
&=& {\bm W}\cdot {\bm M}({\bm x}, {\bm y}) \\
&=& {\rm vec}({\bm W})^H{\rm vec}({\bm M}({\bm x}, {\bm y}))
\end{eqnarray*}
and
$${\rm \mathbb{E}}\big[ \bar{f}({\bm W},{\bm x},{\bm y})^2  \big]={\rm vec}({\bm W})^H{\bm U}{\rm vec}({\bm W}),$$
where
$${\bm U}={\rm \mathbb{E}}\big[{\rm vec}({\bm M}({\bm x}, {\bm y})){\rm vec}({\bm M}({\bm x}, {\bm y}))^H\big].$$
The explicit forms of ${\bm M}({\bm x}, {\bm y})$ and ${\bm U}$ will be given in the sequel. Armed with Theorem \ref{thm:1}, the following second-order cone constraint serves as a safe approximation of the chance constraint ${\rm Pr} \{ Q( \bm{W}, \bm{x}, \bm{y}) \ge 0 \} \le \rho$ in \eqref{main1}:
$$a_0({\bm W}) \ge c_1(\rho) \| {\bm U}^{1/2} {\rm vec}({\bm W})\|.$$
This yields the following safe approximation of Problem \eqref{main1}:
\begin{equation}\label{main2}
\begin{array}{ll}
\displaystyle\min_{{\bm W} \in \mathbb{C}^{L \times L}} &  \quad \displaystyle (P_t{\bm I}\odot(\bar{\bm f}\bar{\bm f}^H+\epsilon^2{\bm I}) + {\bm \Sigma})\cdot {\bm W}\\
\text{subject to} \quad &\quad \displaystyle a_0({\bm W}) \ge c_1(\rho) \| {\bm U}^{1/2} {\rm vec}({\bm W})\|,\\
& \quad {\bm W} \succeq 0.
\end{array}
\end{equation}
Note that Problem \eqref{main2} can be readily handled by applying the SDR technique and off-the-shelf convex solvers. 
\ifconfver
\begin{table*}[th]
	\else
	\begin{table}[h]
		\fi
		\caption{Explicit Expressions in Fourth-order Moment Inequality-based Approach}
		\label{tab:summary}
			\linespread{1.25} \rm \footnotesize 
			\begin{tabular}{c||l}
			\hline\hline
				${\bm M}({\bm x}, {\bm y})_{(i,i)}$ &
				\begin{minipage}{0.8\textwidth}
					\begin{center}
\begin{eqnarray*}\label{{Mii}}
	&&{\bm M}({\bm x}, {\bm y})_{(i,i)}= \eta\sigma_i^2({y_i}{\bar g}_i^*+{\bar g}_i{y_i}^*) + \eta^2\sigma_i^2{y_i}{y_i}^* - \frac{P_t}{\gamma}\big[
	\big(\eta{\bar f}_i{\bar f}_i^*({y_i}^*{\bar g}_i +
	{\bar g}_i^*{y_i}) \\
	&&+\epsilon{\bar g}_i{\bar g}_i^*({x_i}{\bar f}_i^*+{\bar f}_i{x_i}^*)\big)
	+\big(\eta^2{\bar f}_i{\bar f}_i^*{y_i}{y_i}^* + \epsilon^2 {\bar g}_i{\bar g}_i^*{x_i}{x_i}^* + \epsilon\eta({x_i}{\bar f}_i^*+{\bar f}_i{x_i}^*)({y_i}^*{\bar g}_i+{\bar g}_i^*{y_i})\big) \\
	&&+ \big(\epsilon\eta^2({x_i}{\bar f}_i^*+{\bar f}_i{x_i}^*){y_i}{y_i}^*+\eta\epsilon^2{x_i}{x_i}^*({y_i}^*{\bar g}_i+{\bar g}_i^*{y_i})\big) 
	+ \epsilon^2\eta^2{x_i}{x_i}^*{y_i}{y_i}^* 
	\big]\\
\end{eqnarray*}
					\end{center}
				\end{minipage}
			 \\ \hline\hline
${\bm M}({\bm x}, {\bm y})_{(k,\ell)}$ &
			 \begin{minipage}{0.8\textwidth}
			 	\begin{center}
			 		\begin{eqnarray*}\label{{Mij}}
			 			&&{\bm M}({\bm x}, {\bm y})_{(k,\ell)}=-\frac{P_t}{\gamma}\big[
			 			\big(\eta{\bar f}_k{\bar f}_{\ell}^*({y_k}^*{\bar g}_{\ell} +{\bar g}_k^*{y_{\ell}}) +\epsilon{\bar g}_k^*{\bar g}_{\ell}({x_k}{\bar f}_{\ell}^*+{\bar f}_k{x_{\ell}}^*)\big)\\
			 			&&+\big(\eta^2{{\bar f}_k}{{\bar f}_{\ell}}^*{y_k}^*{y_{\ell}} + \epsilon^2 {\bar g}_k^*{\bar g}_{\ell}{x_k}{x_{\ell}}^* + \epsilon\eta({x_k}{\bar f}_{\ell}^*+{\bar f}_k{x_{\ell}}^*)({y_k}^*{\bar g}_{\ell}+{\bar g}_k^*{y_{\ell}})\big) \\
			 			&&+ \big(\epsilon\eta^2({x_k}{\bar f}_{\ell}^*+{\bar f}_k{x_{\ell}}^*){y_k}^*{y_{\ell}}+\eta\epsilon^2{x_k}{x_{\ell}}^*({y_k}^*{\bar g}_{\ell}+{\bar g}_k^*{y_{\ell}})\big) 
			 			+ \epsilon^2\eta^2{x_k}{x_{\ell}}^*{y_k}^*{y_{\ell}} 
			 			\big], \quad 1 \le i\le L, 1 \le k \neq \ell \le L\\
			 		\end{eqnarray*}
			 	\end{center}
			 \end{minipage}
			 \\ \hline\hline
$ {\bm U}_{(i,j),(k,\ell)}$ &
\begin{minipage}{0.8\textwidth}
	\begin{center}			 	 
			 \begin{equation*}
			 {\bm U}_{(i,j),(k,\ell)}=\left\{
			 \begin{array}{rl}
			 \Sigma_{m=1}^9 C_m(i)C_m(i)^*    &\ {\rm if}~ 1 \le i=j=k=\ell \le L;\\
			 C_9(i)C_9(k)^*&\ {\rm if}\ i=j,\  k=\ell, \  i \neq k;\\
			 C_1(i)D_1(i,\ell)^* + C_3(i)D_3(i,\ell)^* + C_5(i)D_5(i,\ell)^* &\ {\rm if}\ i=j, \ k\neq\ell,\ i=k;\\
			 C_2(i)D_2(k,i)^* + C_4(i)D_4(k,i)^* + C_8(i)D_8(k,i)^*  &\ {\rm if}\ i=j, \ k\neq\ell,\ i=\ell;\\
			 D_1(i,j)C_1(i)^*+D_3(i,j)C_3(i)^* + D_5(i,j)C_5(i)^* &\ {\rm if}\ i\neq j, \ k=\ell,\ i=k;\\
			 D_2(i,j)C_2(j)^*+D_4(i,j)C_4(j)^* + D_8(i,j)C_8(j)^*&\ {\rm if}\ i\neq j, \ k=\ell,\ j=k;\\
			 \Sigma_{n=1}^{15} D_n(i,j)D_n(i,j)^*    &\ {\rm if}~i\neq j,\ k\neq \ell,\ i=k,\ j=\ell;\\
			 D_1(k,j)D_1(k,\ell)^*+D_3(k,j)D_3(k,\ell)^* + D_5(k,j)D_5(k,\ell)^* &\ {\rm if}~ i\neq j,\ k\neq \ell,\ i=k,\ j\neq\ell;\\
			 D_2(i,\ell)D_2(k,\ell)^*+D_4(i,\ell)D_4(k,\ell)^* + D_8(i,\ell)D_8(k,\ell)^* &\ {\rm if}~ i\neq j,\ k\neq \ell,\ i\neq k,\ j=\ell;\\
			 0\ &\ {\rm otherwise.}\
			 \end{array}
			 \right. 
			 \end{equation*}
		\end{center}
		 \end{minipage}
			 \\ \hline\hline
			\end{tabular}
		\ifconfver
	\end{table*}
	\else
\end{table}
\fi

A straightforward but tedious calculation yields ${\bm M}({\bm x}, {\bm y})_{(i,j)}$ in Table~\ref{tab:summary}. To derive an explicit expression for ${\bm U}$, we treat $(i,j)$ as an index of any vectorized matrix in the sequel. In other words, for a matrix ${\bm R}$, we denote
$${\rm vec}({\bm R})_{(i,j)}={\bm R}_{i,j}.$$
Recall that ${\bm U}={\rm \mathbb{E}}\big[{\rm vec}({\bm M}({\bm x}, {\bm y})){\rm vec}({\bm M}({\bm x}, {\bm y}))^H\big].$
Then,
\begin{eqnarray*}\label{U}
{\bm U}_{(i,j),(k,\ell)}&=&{\rm \mathbb{E}}\big[{\rm vec}({\bm M}({\bm x}, {\bm y}))_{(i,j)}{\rm vec}({\bm M}({\bm x}, {\bm y}))^*_{(k,\ell)}\big]\\
&=&{\rm \mathbb{E}}\big[{\bm M}({\bm x}, {\bm y})_{i,j}{\bm M}({\bm x}, {\bm y})^*_{k,\ell}\big].
\end{eqnarray*}
The entries of $\bm U$ are given as ${\bm U}_{(i,j),(k,\ell)}$ in Table \ref{tab:summary} with
\begin{equation*}
\left\{
\begin{array}{rl}
&C_1(i)=\eta\sigma_i^2{\bar g}_i - \frac{P_t}{\gamma}(\eta{\bar f}_i{\bar f}_i^*{\bar g}_i+\eta\epsilon^2{\bar g}_i) ;\\
&C_2(i)=\eta\sigma_i^2{\bar g}_i^* - \frac{P_t}{\gamma}(\eta{\bar f}_i{\bar f}_i^*{\bar g}_i^*+\eta\epsilon^2{\bar g}_i^*) ;\\
&C_3(i)= - \frac{P_t}{\gamma}(\epsilon{\bar g}_i{\bar g}_i^*{\bar f}_i^*+\epsilon\eta^2{\bar f}_i^*) ;\\
&C_4(i)= - \frac{P_t}{\gamma}(\epsilon{\bar g}_i{\bar g}_i^*{\bar f}_i+\epsilon\eta^2{\bar f}_i) ;\\
&C_5(i)= - \frac{P_t}{\gamma}\epsilon\eta{\bar f}_i^*{\bar g}_i;\\
&C_6(i)= - \frac{P_t}{\gamma}\epsilon\eta{\bar f}_i^*{\bar g}_i^*;\\
&C_7(i)= - \frac{P_t}{\gamma}\epsilon\eta{\bar f}_i{\bar g}_i;\\
&C_8(i)= - \frac{P_t}{\gamma}\epsilon\eta{\bar f}_i{\bar g}_i^*;\\
&C_9(i)= \eta^2\sigma_i^2 - \frac{P_t}{\gamma}(\eta^2{\bar f}_i{\bar f}_i^*+\epsilon^2{\bar g}_i{\bar g}_i^*+\epsilon^2\eta^2)
\end{array}
\right.
\end{equation*}
and
\begin{equation*}
\left\{
\begin{array}{rl}
&D_1(k,\ell)= -\frac{P_t}{\gamma}\eta{\bar f}_k{\bar f}_{\ell}^*{\bar g}_{\ell} ;\\ 
&D_2(k,\ell)= -\frac{P_t}{\gamma}\eta{\bar f}_k{\bar f}_{\ell}^*{\bar g}_k^* ;\\ 
&D_3(k,\ell)= -\frac{P_t}{\gamma}\epsilon{\bar g}_k^*{\bar g}_{\ell}{\bar f}_{\ell}^* ;\\ 
&D_4(k,\ell)= -\frac{P_t}{\gamma}\epsilon{\bar g}_k^*{\bar g}_{\ell}{\bar f}_k ;\\ 
&D_5(k,\ell)= -\frac{P_t}{\gamma}\epsilon\eta{\bar f}_{\ell}^*{\bar g}_{\ell} ;\\ 
&D_6(k,\ell)= -\frac{P_t}{\gamma}\epsilon\eta{\bar f}_{\ell}^*{\bar g}_{k}^*;\\ 
&D_7(k,\ell)= -\frac{P_t}{\gamma}\epsilon\eta{\bar f}_{k}{\bar g}_{\ell} ;\\ 
&D_8(k,\ell)= -\frac{P_t}{\gamma}\epsilon\eta{\bar f}_{k}{\bar g}_{k}^* ;\\ 
&D_9(k,\ell)= -\frac{P_t}{\gamma}\eta^2{\bar f}_k{\bar f}_{\ell}^* ;\\ 
&D_{10}(k,\ell)= -\frac{P_t}{\gamma}\epsilon^2{\bar g}_k^*{\bar g}_{\ell} ;\\ 
&D_{11}(k,\ell)= -\frac{P_t}{\gamma}\epsilon\eta^2{\bar f}_{\ell}^* ;\\ 
&D_{12}(k,\ell)= -\frac{P_t}{\gamma}\epsilon\eta^2{\bar f}_{k}  ;\\ 
&D_{13}(k,\ell)= -\frac{P_t}{\gamma}\epsilon^2\eta{\bar g}_{\ell} ;\\ 
&D_{14}(k,\ell)= -\frac{P_t}{\gamma}\epsilon^2\eta{\bar g}_{k}^*  ;\\ 
&D_{15}(k,\ell)= -\frac{P_t}{\gamma}\epsilon^2\eta^2.\\ 
\end{array}
\right.
\end{equation*}

\subsection{The Second-order Moment Inequality Method}
The aforementioned fourth-order moment inequality method can provide a feasible solution to Problem \eqref{main1}. However, as it involves  all orders (from first to fourth) of channel uncertainties in the outage constraint, the ultimate explicit expression is rather complicated. In the robust design, we observe that channel uncertainties are usually small and thus higher-order uncertainty terms contribute very little to the SNR quantity. Hence, we are motivated to ignore the third and fourth-order uncertainty terms so as to significantly simplify our problem. 

To proceed, we let $f_{quad}({\bm W},{\bm x},{\bm y})$ be the quadratic approximation of $f({\bm W},{\bm x},{\bm y})$, obtained by dropping the cubic and quartic terms in $f({\bm W},{\bm x},{\bm y})$. 
Then, we consider the following second-order approximation of the original chance constraint:
\begin{equation}\label{2nd-chance}
\Pr (f_{quad}({\bm W},{\bm x},{\bm y})\ge 0)\le \rho. 
\end{equation}
To tackle the above constraint, we define
\begin{eqnarray*}
\bar{f}_{quad}({\bm W},{\bm x},{\bm y})&=& f_{quad}({\bm W},{\bm x},{\bm y})+a_0({\bm W})\\
&=& {\bm W}\cdot {\bm M}_{quad}({\bm x}, {\bm y}) \\
&=& {\rm vec}({\bm W})^H{\rm vec}({\bm M}_{quad}({\bm x}, {\bm y}))
\end{eqnarray*}
and
${\rm \mathbb{E}}\big[ \bar{f}_{quad}({\bm W},{\bm x},{\bm y})^2  \big]={\rm vec}({\bm W})^H{\bm U}_{quad}{\rm vec}({\bm W}),$
where
${\bm U}_{quad}={\rm \mathbb{E}}\big[{\rm vec}({\bm M}_{quad}({\bm x}, {\bm y})){\rm vec}({\bm M}_{quad}({\bm x}, {\bm y}))^H\big]$.
The explicit forms of ${\bm M}_{quad}({\bm x}, {\bm y})$ and $\bm U_{quad}$ are given in the sequel.
By Theorem 1, the second-order cone constraint
$$a_0({\bm W}) \ge c_2(\rho) \| {\bm U_{quad}}^{1/2} {\rm vec}({\bm W})\|$$
serves as a safe approximation of the chance constraint \eqref{2nd-chance}.
This yields the following safe approximation of Problem \eqref{main1}:
\begin{equation}\label{main2_quad}
\begin{array}{ll}
\displaystyle\min_{{\bm W} \in \mathbb{C}^{L \times L}} &  \quad \displaystyle (P_t{\bm I}\odot(\bar{\bm f}\bar{\bm f}^H+\epsilon^2{\bm I}) + {\bm \Sigma})\cdot {\bm W}\\
\text{subject to} \quad &\quad \displaystyle a_0({\bm W}) \ge c_2(\rho) \| {\bm U}_{quad}^{1/2} {\rm vec}({\bm W})\|,\\
& \quad {\bm W} \succeq 0.
\end{array}
\end{equation}
Note that Problem \eqref{main2_quad} can be readily handled by applying SDR technique and off-the-shelf convex solvers.

The explicit expression for $\bm M_{quad}({\bm x}, {\bm y})$ and $\bm U_{quad}$ can be obtained in a similar way. We provide the explicit expression of $\bm M_{quad}({\bm x}, {\bm y})$ and $\bm U_{quad}$
in Table \ref{tab:summary1} with
\begin{equation*}\label{Coeff_Mii-quad}
\left\{
\begin{array}{rl}
&C_1(i)=\eta\sigma_i^2{\bar g}_i - \frac{P_t}{\gamma}\eta{\bar f}_i{\bar f}_i^*{\bar g}_i ;\\ 
&C_2(i)=\eta\sigma_i^2{\bar g}_i^* - \frac{P_t}{\gamma}\eta{\bar f}_i{\bar f}_i^*{\bar g}_i^* ;\\  
&C_3(i)= - \frac{P_t}{\gamma}\epsilon{\bar g}_i{\bar g}_i^*{\bar f}_i^* ;\\ 
&C_4(i)= - \frac{P_t}{\gamma}\epsilon{\bar g}_i{\bar g}_i^*{\bar f}_i ;\\ 
&C_5(i)= - \frac{P_t}{\gamma}\epsilon\eta{\bar f}_i^*{\bar g}_i;\\ 
&C_6(i)= - \frac{P_t}{\gamma}\epsilon\eta{\bar f}_i^*{\bar g}_i^*;\\ 
&C_7(i)= - \frac{P_t}{\gamma}\epsilon\eta{\bar f}_i{\bar g}_i;\\ 
&C_8(i)= - \frac{P_t}{\gamma}\epsilon\eta{\bar f}_i{\bar g}_i^*;\\ 
&C_9(i)= \eta^2\sigma_i^2 - \frac{P_t}{\gamma}(\eta^2{\bar f}_i{\bar f}_i^*+\epsilon^2{\bar g}_i{\bar g}_i^*)
\end{array}
\right.
\end{equation*}
and
\begin{equation*}\label{Coeff_Mij-quad}
\left\{
\begin{array}{rl}
&D_1(k,\ell)= -\frac{P_t}{\gamma}\eta{\bar f}_k{\bar f}_{\ell}^*{\bar g}_{\ell} ;\\ 
&D_2(k,\ell)= -\frac{P_t}{\gamma}\eta{\bar f}_k{\bar f}_{\ell}^*{\bar g}_k^* ;\\ 
&D_3(k,\ell)= -\frac{P_t}{\gamma}\epsilon{\bar g}_k^*{\bar g}_{\ell}{\bar f}_{\ell}^* ;\\ 
&D_4(k,\ell)= -\frac{P_t}{\gamma}\epsilon{\bar g}_k^*{\bar g}_{\ell}{\bar f}_k ;\\ 
&D_5(k,\ell)= -\frac{P_t}{\gamma}\epsilon\eta{\bar f}_{\ell}^*{\bar g}_{\ell} ;\\ 
&D_6(k,\ell)= -\frac{P_t}{\gamma}\epsilon\eta{\bar f}_{\ell}^*{\bar g}_{k}^*;\\ 
&D_7(k,\ell)= -\frac{P_t}{\gamma}\epsilon\eta{\bar f}_{k}{\bar g}_{\ell} ;\\ 
&D_8(k,\ell)= -\frac{P_t}{\gamma}\epsilon\eta{\bar f}_{k}{\bar g}_{k}^* ;\\ 
&D_9(k,\ell)= -\frac{P_t}{\gamma}\eta^2{\bar f}_k{\bar f}_{\ell}^* ;\\ 
&D_{10}(k,\ell)= -\frac{P_t}{\gamma}\epsilon^2{\bar g}_k^*{\bar g}_{\ell}.\\ 
\end{array}
\right.
\end{equation*}
Apparently, the above second-order moment inequality-based approach involves less perturbation terms compared to its fourth-order counterpart. In fact, compared to Problem \eqref{main2}, it is easier to find a feasible solution to Problem \eqref{main2_quad}.


\ifconfver
\begin{table*}[th]
	\else
	\begin{table}[h]
		\fi
		\caption{Explicit Expressions in Second-order Moment Inequality-based Approach}
		\label{tab:summary1}
		\linespread{1.25} \rm \footnotesize 
		\begin{tabular}{c||l}
			\hline\hline
			${\bm M}({\bm x}, {\bm y})_{(i,i)}$ &
			\begin{minipage}{0.8\textwidth}
				\begin{center}
\begin{eqnarray*}\label{{Mii-quad}}
	&&{\bm M}({\bm x}, {\bm y})_{(i,i)}= \eta\sigma_i^2({y_i}{\bar g}_i^*+{\bar g}_i{y_i}^*) + \eta^2\sigma_i^2{y_i}{y_i}^* - \frac{P_t}{\gamma}\big[
	\eta{\bar f}_i{\bar f}_i^*({y_i}^*{\bar g}_i +
	{\bar g}_i^*{y_i}) +\epsilon{\bar g}_i{\bar g}_i^*({x_i}{\bar f}_i^*+{\bar f}_i{x_i}^*)\\
	&&+\eta^2{\bar f}_i{\bar f}_i^*{y_i}{y_i}^* + \epsilon^2 {\bar g}_i{\bar g}_i^*{x_i}{x_i}^* + \epsilon\eta({x_i}{\bar f}_i^*+{\bar f}_i{x_i}^*)({y_i}^*{\bar g}_i+{\bar g}_i^*{y_i})\big] 
\end{eqnarray*}
				\end{center}
			\end{minipage}
			\\ \hline\hline
			${\bm M}({\bm x}, {\bm y})_{(k,\ell)}$ &
			\begin{minipage}{0.8\textwidth}
				\begin{center}
\begin{eqnarray*}\label{{Mij-quad}}
	&&{\bm M}({\bm x}, {\bm y})_{(k,\ell)}=-\frac{P_t}{\gamma}\big[
	\eta{\bar f}_k{\bar f}_{\ell}^*({y_k}^*{\bar g}_{\ell} +{\bar g}_k^*{y_{\ell}}) +\epsilon{\bar g}_k^*{\bar g}_{\ell}({x_k}{\bar f}_{\ell}^*+{\bar f}_k{x_{\ell}}^*)\\
	&&+\eta^2{{\bar f}_k}{{\bar f}_{\ell}}^*{y_k}^*{y_{\ell}} + \epsilon^2 {\bar g}_k^*{\bar g}_{\ell}{x_k}{x_{\ell}}^* + \epsilon\eta({x_k}{\bar f}_{\ell}^*+{\bar f}_k{x_{\ell}}^*)({y_k}^*{\bar g}_{\ell}+{\bar g}_k^*{y_{\ell}})\big], \quad 
1 \le i\le L, 1 \le k \neq \ell \le L
\end{eqnarray*}
				\end{center}
			\end{minipage}
			\\ \hline\hline
			${\bm U}^{quad}_{(i,j),(k,\ell)}$ &
			\begin{minipage}{0.8\textwidth}
				\begin{center}			 	 
\begin{equation*}\label{Uijkl-quad}
{\bm U}^{quad}_{(i,j),(k,\ell)}=\left\{
\begin{array}{rl}
\Sigma_{m=1}^9 C_m(i)C_m(i)^*    &\ {\rm if}~ 1 \le i=j=k=\ell \le L;\\
C_9(i)C_9(k)^*&\ {\rm if}\ i=j,\  k=\ell, \  i \neq k;\\
C_1(i)D_1(i,\ell)^* + C_3(i)D_3(i,\ell)^* + C_5(i)D_5(i,\ell)^* &\ {\rm if}\ i=j, \ k\neq\ell,\ i=k;\\
C_2(i)D_2(k,i)^* + C_4(i)D_4(k,i)^* + C_8(i)D_8(k,i)^*  &\ {\rm if}\ i=j, \ k\neq\ell,\ i=\ell;\\
D_1(i,j)C_1(i)^*+D_3(i,j)C_3(i)^* + D_5(i,j)C_5(i)^* &\ {\rm if}\ i\neq j, \ k=\ell,\ i=k;\\
D_2(i,j)C_2(j)^*+D_4(i,j)C_4(j)^* + D_8(i,j)C_8(j)^*&\ {\rm if}\ i\neq j, \ k=\ell,\ j=k;\\
\Sigma_{n=1}^{10} D_n(i,j)D_n(i,j)^*    &\ {\rm if} \ i\neq j,\ k\neq \ell,\ i=k,\ j=\ell;\\
D_1(k,j)D_1(k,\ell)^*+D_3(k,j)D_3(k,\ell)^* + D_5(k,j)D_5(k,\ell)^* &\ {\rm if}\ i\neq j,\ k\neq \ell,\ i=k,\ j\neq\ell;\\
D_2(i,\ell)D_2(k,\ell)^*+D_4(i,\ell)D_4(k,\ell)^* + D_8(i,\ell)D_8(k,\ell)^* &\ {\rm if}\ i\neq j,\ k\neq \ell,\ i\neq k,\ j=\ell;\\
0\ &\ {\rm otherwise.}\
\end{array}
\right.
\end{equation*}
				\end{center}
			\end{minipage}
			\\ \hline\hline
		\end{tabular}
		\ifconfver
	\end{table*}
	\else
\end{table}
\fi

\section{The Bernstein-type Inequality Approach} \label{sec:4}
Although the fourth-order moment inequality-based approach is much more restricted and the feasibility rate can be lower, it can provide a more robust AF design, compared to the second-order counterpart. Hence, in practice, we can first try to find an AF weight via Problem \eqref{main2}. If it cannot provide a feasible solution, we then drop the higher-order perturbations and resort to a quadratically perturbed chance constraint problem, e.g., Problem \eqref{main2_quad}, to find a relatively good AF weight. In this section, we introduce another typical way of tackling the quadratically perturbed chance constraint, i.e., to use the so-called Bernstein-type inequality \cite{wang2014outage}. 

To be specific, we define ${\bm \xi}=\sqrt{2}\left(\begin{array}{c}{\rm Re}({\bm x})\\{\rm Im}({\bm x})\\{\rm Re}({\bm y})\\{\rm Im}({\bm y}) \end{array}\right)\sim\mathcal{N}(0,{\bm {I}}_{4L})$, with ${\rm Re}(\cdot)$ and ${\rm Im}(\cdot)$ denoting the real part and imaginary part of a complex number, respectively. Then, the quadratic approximation can be rewritten as
\begin{align}\label{SNR_B2}
{f}_{quad}({\bm \xi}) = s_0({\bm \xi})+s_1({\bm \xi})+s_2({\bm \xi}),
\end{align}
where $s_0({\bm \xi}), s_1({\bm \xi}), s_2({\bm \xi})$ represent the constant term, linear term, and quadratic term in ${f}_{quad}({\bm \xi})$, respectively. More precisely, we have
\begin{align*}
&s_0({{\bm \xi}})\\
= &- \sigma_{{v}}^2 - {\bm W}\cdot \big({\bm \Sigma}\odot (\bar{\bm g}\bar{\bm g}^H)\big) + \frac{P_t}{\gamma} {\bm W}\cdot \big({\bar {\bm f}}{\bar{\bm f}}^H \odot ({\bar{\bm g}}^*)({\bar{ \bm g}}^*)^H\big), \\
&s_1({{\bm \xi}}) \\
= &{\bm W}\cdot\big(
\frac{P_t \epsilon}{\gamma} ({\bm x}\bar{\bm f}^H+ \bar{\bm f}{\bm x}^H)\odot{\bar{\bm g}}^*({\bar{ \bm g}}^*)^H-\eta{\bm \Sigma}\odot (\bar{\bm g}{\bm y^H}+{\bm y}\bar{\bm g}^H)\big)\\
& +\frac{P_t\eta}{\gamma} {\bm W}\cdot\big((\bar{\bm f}\bar{\bm f}^H)\odot(({\bm y}^*)(\bar{\bm g}^*)^H+(\bar{\bm g}^*)({\bm y}^*)^H), \\
&s_2({{\bm \xi}}) \\
= &\frac{P_t}{\gamma} {\bm W}\cdot \big(\epsilon^2({\bm x}{\bm x}^H)\odot({\bar{\bm g}}^*)({\bar{ \bm g}}^*)^H  +\eta^2(\bar{\bm f}\bar{\bm f}^H)\odot(({\bm y}^*)({\bm y}^*)^H)\big) \\
&+\frac{P_t}{\gamma} {\bm W}\cdot \big(\eta\epsilon ({\bm x}\bar{\bm f}^H+ \bar{\bm f}{\bm x}^H)\odot(({\bm y}^*)(\bar{\bm g}^*)^H+(\bar{\bm g}^*)({\bm y}^*)^H)
 \big)\\
&-\eta^2{\bm W}\cdot\big({\bm \Sigma}\odot ({\bm y}{\bm y}^H)\big).
\end{align*}
Define ${\bm F}=\bar{\bm f}\bar{\bm f}^H$, ${\bm G}=(\bar{\bm g}^*)(\bar{\bm g}^*)^H$, one can easily verify that $$s_1({{\bm \xi}})= {\bm v}_{{\bm \xi}}^T{\bm \xi},$$where
$${\bm v}_{{\bm \xi}}={\sqrt{2}}\left(\begin{array}{c}
\epsilon\frac{Pt}{\gamma}{\rm Re}\big({\bm W}\odot{\bm G}^*)\bar{\bm f}\big)\\
\epsilon\frac{Pt}{\gamma}{\rm Im}\big({\bm W}\odot{\bm G}^*)\bar{\bm f}\big)\\
- \eta {\rm Re}({\bm W}\odot {\bm \Sigma}){\rm Re}(\bar{\bm g}) + \eta\frac{Pt}{\gamma}{\rm Re}\big({\bm W}^*\odot{\bm F})\bar{\bm g}\big) \\
- \eta {\rm Re}({\bm W}\odot {\bm \Sigma}){\rm Im}(\bar{\bm g}) + \eta\frac{Pt}{\gamma}{\rm Im}\big({\bm W}^*\odot{\bm F})\bar{\bm g}\big) \\
\end{array}
\right).$$ In addition, consider
\begin{align*}\label{s2}
{\bm R}=\frac{1}{2}\left(
\begin{array}{cc}
\epsilon^2 \frac{P_t}{\gamma} {\bm K_1} & \epsilon\eta\frac{P_t}{\gamma}{\bm K_2}\\
\epsilon\eta\frac{P_t}{\gamma}{\bm K_3} & -\eta^2{\bm K_4^{(1)}} + \eta^2 \frac{P_t}{\gamma}{\bm K_4^{(2)}}\\
\end{array}
\right),
\end{align*}
where
\begin{align*}
&{\bm K_1} = \left(
\begin{array}{cc}
{\rm Re}({\bm W}^*\odot{\bm G}) & {\rm Im}({\bm W}^*\odot{\bm G})\\
-{\rm Im}({\bm W}^*\odot{\bm G}) & {\rm Re}({\bm W}^*\odot{\bm G})\\
\end{array}
\right),\\
&{\bm K_4^{(1)}} =  \left(
\begin{array}{cc}
{\rm Re}({\bm W}\odot{\bm \Sigma}) & {\bm 0}_{L\times L}\\
{\bm 0}_{L\times L} & {\rm Re}({\bm W}\odot{\bm \Sigma})\\
\end{array}
\right),\\
&{\bm K_4^{(2)}} = \left(
\begin{array}{cc}
{\rm Re}({\bm W}\odot{\bm F}^*) & {\rm Im}({\bm W}\odot{\bm F}^*)\\
-{\rm Im}({\bm W}\odot{\bm F}^*) & {\rm Re}({\bm W}\odot{\bm F}^*)\\
\end{array}
\right), \\
\end{align*}
\ifconfver
and ${\bm K_2}, {\bm K_3}$ are given at the top of the next page.
\begin{figure*}[!btp]
	\normalsize
	\begin{align*}
	&{\bm K_2} = \left(
	\begin{array}{cc}
	{\rm Diag}\big({\rm Re}({\bm W}(\bar{\bm f}\odot\bar{\bm g}^*))\big) + {\rm Re}\big({\bm W}\odot(\bar{\bm g}\bar{\bm f}^T)\big)& 
	{\rm Diag}\big({\rm Im}({\bm W}^*(\bar{\bm f}^*\odot\bar{\bm g}))\big) + {\rm Im}\big({\bm W}\odot(\bar{\bm g}\bar{\bm f}^T)\big) \\
	{\rm Diag}\big({\rm Im}({\bm W}(\bar{\bm f}\odot\bar{\bm g}^*))\big) + {\rm Im}\big({\bm W}\odot(\bar{\bm g}\bar{\bm f}^T)\big)&
	{\rm Diag}\big({\rm Re}({\bm W}^*(\bar{\bm f}^*\odot\bar{\bm g}))\big) - {\rm Re}\big({\bm W}\odot(\bar{\bm g}\bar{\bm f}^T)\big) \\
	\end{array}
	\right)\\
	&{\bm K_3} = \left(
	\begin{array}{cc}
	{\rm Diag}\big({\rm Re}({\bm W}(\bar{\bm f}\odot\bar{\bm g}^*))\big) + {\rm Re}\big({\bm W}^*\odot(\bar{\bm f}\bar{\bm g}^T)\big)& 
	{\rm Diag}\big({\rm Im}({\bm W}(\bar{\bm f}\odot\bar{\bm g}^*))\big) + {\rm Im}\big({\bm W}^*\odot(\bar{\bm f}\bar{\bm g}^T)\big) 
	\\
	{\rm Diag}\big({\rm Im}({\bm W}^*(\bar{\bm f}^*\odot\bar{\bm g}))\big) + {\rm Im}\big({\bm W}^*\odot(\bar{\bm f}\bar{\bm g}^T)\big)&
	{\rm Diag}\big({\rm Re}({\bm W}^*(\bar{\bm f}^*\odot\bar{\bm g}))\big) - {\rm Re}\big({\bm W}^*\odot(\bar{\bm f}\bar{\bm g}^T)\big)
	\\
	\end{array}
	\right)\\
	\end{align*}
	\hrulefill
	\vspace*{4pt}
\end{figure*}
\else
and ${\bm K_2}, {\bm K_3}$ are given as follows:
\begin{align*}
&{\bm K_2} = \left(
\begin{array}{cc}
	{\rm Diag}\big({\rm Re}({\bm W}(\bar{\bm f}\odot\bar{\bm g}^*))\big) + {\rm Re}\big({\bm W}\odot(\bar{\bm g}\bar{\bm f}^T)\big)& 
	{\rm Diag}\big({\rm Im}({\bm W}^*(\bar{\bm f}^*\odot\bar{\bm g}))\big) + {\rm Im}\big({\bm W}\odot(\bar{\bm g}\bar{\bm f}^T)\big) \\
	{\rm Diag}\big({\rm Im}({\bm W}(\bar{\bm f}\odot\bar{\bm g}^*))\big) + {\rm Im}\big({\bm W}\odot(\bar{\bm g}\bar{\bm f}^T)\big)&
	{\rm Diag}\big({\rm Re}({\bm W}^*(\bar{\bm f}^*\odot\bar{\bm g}))\big) - {\rm Re}\big({\bm W}\odot(\bar{\bm g}\bar{\bm f}^T)\big) \\
\end{array}
\right)\\
&{\bm K_3} = \left(
\begin{array}{cc}
	{\rm Diag}\big({\rm Re}({\bm W}(\bar{\bm f}\odot\bar{\bm g}^*))\big) + {\rm Re}\big({\bm W}^*\odot(\bar{\bm f}\bar{\bm g}^T)\big)& 
	{\rm Diag}\big({\rm Im}({\bm W}(\bar{\bm f}\odot\bar{\bm g}^*))\big) + {\rm Im}\big({\bm W}^*\odot(\bar{\bm f}\bar{\bm g}^T)\big) 
	\\
	{\rm Diag}\big({\rm Im}({\bm W}^*(\bar{\bm f}^*\odot\bar{\bm g}))\big) + {\rm Im}\big({\bm W}^*\odot(\bar{\bm f}\bar{\bm g}^T)\big)&
	{\rm Diag}\big({\rm Re}({\bm W}^*(\bar{\bm f}^*\odot\bar{\bm g}))\big) - {\rm Re}\big({\bm W}^*\odot(\bar{\bm f}\bar{\bm g}^T)\big)
	\\
\end{array}
\right)\\
\end{align*}
\fi
We symmetrize ${\bm R}$ by setting $\bar{{\bm R}} = \frac{1}{2}({{\bm R}}+{{\bm R}}^T)$, so that $s_2({{\bm \xi}})= {\bm \xi}^T \bar{{\bm R}} {\bm \xi}$. Furthermore, denote $s_{{\bm \xi}}=s_0({\bm \xi})$, so that we can rewrite the quadratic approximation as
$${f}_{quad}({\bm \xi}) = s_{{\bm \xi}}+{\bm v}_{{\bm \xi}}^T{\bm \xi}+{\bm \xi}^T \bar{\bm R} {\bm \xi}.$$
Now, we apply the  Bernstein-type inequality to the chance constraint
$${\rm Pr} \{s_{{\bm \xi}}+{\bm v}_{{\bm \xi}}^T{\bm \xi}+{\bm \xi}^T \bar{{\bm R}} {\bm \xi}\ge 0\} \ge 1-\rho$$
and obtain the following safe approximation of our problem:
\begin{equation}\label{main3}
\begin{array}{ll}
\displaystyle\min_{{\bm W}, \lambda, \delta
} &
 \displaystyle (P_t{\bm I}\odot(\bar{\bm f}\bar{\bm f}^H+\epsilon^2{\bm I}) + {\bm \Sigma})\cdot {\bm W}\\
\text{subject to}  &\displaystyle {\rm Tr}(\bar{{\bm R}}) - 2\sqrt{{\rm ln}(1/\rho)}\cdot \delta + 2{\rm ln}(\rho)\cdot \lambda + s_{\bm \xi} \ge 0, \\
& \sqrt{\|\bar{{\bm R}}\|^2_F + \frac{1}{2}\|{\bm v}_{\bm \xi}\|^2} \le \delta, \\
& \lambda {\bm I}_{4L} + \bar{{\bm R}} \succeq {\bm 0}, \\
& \lambda \ge 0,\\
& {\bm W} \succeq 0.
\end{array}
\end{equation}
Again, the problem can be readily solved by the SDR technique and convex solvers.

{\it{Remark 1:}} For both the moment inequality-based and
Bernstein-type inequality-based AF design, we target at solving Problem \eqref{main1} to find a good AF weight matrix ${\bm W}$, rather than an AF weight vector ${\bm w}$, since we have relaxed the rank constraint in \eqref{main0} for the sake of computational tractability.  Hence, after we have solved \eqref{main1}, we still need to apply the Gaussian randomization algorithm \cite{luo2010semidefinite}; more specifically, see Algorithm 1 in \cite{wu2015relaysbf} on how to find an approximate rank-one AF weight vector ${\bm w}$.

{\it{Remark 2:}} It is worth noting that solving Problem \eqref{main2} yields a safe approximate solution to Problem \eqref{main1}, while solving Problem \eqref{main2_quad} or Problem \eqref{main3} yields only approximate but not necessarily safe solution to Problem \eqref{main1}. This is because we drop the higher-order perturbation terms from the SNR expression.

\section{A Relative Tightness Analysis}\label{sec:5}
When we tackle the second-order approximation of the chance constraint in \eqref{chance_Q}---i.e., $\Pr (f_{quad}({\bm W},{\bm x},{\bm y})\ge 0)\le \rho $---we resort to two types of restriction approaches. The first is the second-order moment inequality approach in \eqref{main2_quad} and the other is the Bernstein-type inequality approach in \eqref{main3}.  In this section, we provide a theoretical relative tightness result for the two approaches.
To proceed, we consider a general Gaussian quadratic polynomial
\begin{eqnarray*}
	Q({\bm \xi}) &=& \sum_{i=1}^m \xi_ia_i + \sum_{i,j=1}^m \xi_i\xi_ja_{ij} = {\bm \xi}^T{\bm A}{\bm \xi} + {\bm a}^T{\bm \xi},
\end{eqnarray*}
where ${\bm \xi}\sim \mathcal{N}(\bz, \bm I)$. Then, our tightness result can be summarized in the following theorem:
\begin{thm} \label{thm:2}
For the chance constraint
\begin{equation} \label{eq:ch-quad}
\Pr( Q({\bm \xi}) \ge t) \le \rho, 
\end{equation}
consider its two safe approximations:
\begin{itemize}
	\item \textbf{Moment inequality-based safe approximation:}
\begin{equation*}
t\ge c(\rho)\sqrt{\E{Q({\bm\xi})^2}}.
\end{equation*}	

	\item \textbf{Bernstein-type inequality-based safe approximation:}
\begin{align*}
&\displaystyle -{\rm Tr}({{\bm A}}) - 2\sqrt{{\rm ln}(1/\rho)}\cdot \delta + 2{\rm ln}(\rho)\cdot \lambda +t \ge 0, \\
& \sqrt{\|{{\bm A}}\|^2_F + \frac{1}{2}\|{\bm a}\|^2} \le \delta, \\
& \lambda {\bm I}_{4L} - {{\bm A}} \succeq {\bm 0}, \\
& \lambda \ge 0.
\end{align*}	
\end{itemize}
Then, for any $\rho \in (\exp(-8),0.00045)$, the moment inequality-based safe approximation is always more conservative than the Bernstein-type inequality-based safe approximation.	
\end{thm}

\begin{proof}
We assume first that $\rho>\exp(-8)$ in the chance constraint \eqref{eq:ch-quad}.  
Using the fact that $\E{\xi_i\xi_j} = \delta_{ij}$, $\E{\xi_i\xi_j\xi_k} = 0$, and
\begin{eqnarray*}
	\E{\xi_i\xi_j\xi_k\xi_l} &=& \left\{
	\begin{array}{c@{\quad}l}
		3 & \mbox{if } i=j=k=l, \\
		1 & \mbox{if there are two distinct indices}, \\
		0 & \mbox{otherwise},
	\end{array}
	\right.
\end{eqnarray*}
we can compute
\begin{eqnarray*}
	\E{Q({\bm \xi})^2} &=& {\bm v}^T \E{{\bm \eta}{\bm \eta}^T} {\bm v} \\
	\noalign{\medskip}
	&=& \sum_{i=1}^m a_i^2 + 3 \sum_{i=1}^m a_{ii}^2 + \sum_{1\le i\not=j\le m} \left( a_{ii}a_{jj} + a_{ij}^2 \right) \\
	\noalign{\medskip}
	&=& \|{\bm a}\|_2^2 + \|{\bm A}\|_F^2 + (\mbox{tr}({\bm A}))^2 + \sum_{i=1}^m a_{ii}^2.
\end{eqnarray*}
Then, following the moment inequality-based approach, we have a safe tractable approximation of~(\ref{eq:ch-quad}):
\begin{align} \label{eq:mom-quad}
t &\ge \sqrt{\frac{\E{Q({\bm\xi})^2}}{\rho}} \\\notag
&= \left( \frac{\|{\bm a}\|_2^2 + \|{\bm A}\|_F^2 + (\mbox{tr}({\bm A}))^2 + \sum_{i=1}^m a_{ii}^2}{\rho} \right)^{1/2}.
\end{align}
On the other hand, the Bernstein-type inequality-based approach yields the following safe tractable approximation of (\ref{eq:ch-quad}):
\begin{equation} \label{eq:bern-quad}
t \ge \mbox{tr}({\bm A}) + 2\sqrt{\ln\frac{1}{\epsilon}}\sqrt{\|{\bm A}\|_F^2 + \frac{1}{2}\|{\bm a}\|_2^2} + \left(2\ln\frac{1}{\epsilon}\right) s^+({\bm A}),
\end{equation}
where $s^+({\bm A}) = \max\{\lambda_{\max}({\bm A}),0\}$.  We claim that for sufficiently small $\rho>\exp(-8)$, every feasible solution to (\ref{eq:mom-quad}) is feasible for (\ref{eq:bern-quad}), which means that (\ref{eq:mom-quad}) is more conservative than (\ref{eq:bern-quad}).  Indeed, if (\ref{eq:mom-quad}) is satisfied, then
\begin{align}
t &\ge \frac{1}{\sqrt{\rho}} \cdot \sqrt{\|{\bm a}\|_2^2 + \|{\bm A}\|_F^2 + (\mbox{tr}({\bm A}))^2} \nonumber \\
&\ge \frac{1}{\sqrt{3\rho}} \left( \|{\bm a}\|_2 + \|A\|_F + |\mbox{tr}({\bm A})| \right) \label{eq:norm-1} \\
&\ge \frac{1}{\sqrt{3\rho}} \left( \left[ \sqrt{ \frac{1}{16}\|{\bm A}\|_F^2 + \|{\bm a}\|_2^2 } + \frac{3}{4} \|{\bm A}\|_F \right] +  |\mbox{tr}({\bm A})|\right) \label{eq:norm-2} \\
&\ge \mbox{tr}({\bm A}) + \frac{1}{4\sqrt{3\rho}} \sqrt{\|{\bm A}\|_F^2 + \frac{1}{2}\|{\bm a}\|_2^2} + \frac{\sqrt{3}}{4\sqrt{\rho}} s^+({\bm A}), \label{eq:eps-bd-1}
\end{align}
where (\ref{eq:norm-1}) and (\ref{eq:norm-2}) follow from the fact that for any ${\bm x} \in \R^n$, 
$$ \|{\bm x}\|_1 \ge \|{\bm x}\|_2 \ge \frac{1}{\sqrt{n}}\|{\bm x}\|_1, $$
and (\ref{eq:eps-bd-1}) holds as long as $\rho<1/3$ (note that $s^+({\bm A}) \le \|{\bm A}\|_F$).  Now, the proof of the claim will be complete if
$$
\frac{1}{4\sqrt{3\rho}} \ge 2\sqrt{\ln\frac{1}{\rho}} \quad\mbox{and}\quad  \frac{\sqrt{3}}{4\sqrt{\rho}} \ge 2\ln\frac{1}{\rho}.
$$
It can be verified that the above inequalities hold if $\rho<0.00045$.  Hence, (\ref{eq:mom-quad}) is more conservative than (\ref{eq:bern-quad}) when $\rho \in (\exp(-8),0.00045)$ (note that $\exp(-8) \approx 0.0003)$. This completes the proof.
\end{proof}
It is worth remarking that a general tightness result for the moment inequality-based and Bernstein-type inequality-based safe approximations is still difficult to obtain, while Theorem \ref{thm:2} provides a provable result in a small region of $\rho$. In the next section, we will provide numerical validations via simulation results.

\begin{figure*}[ht!]
\centering
\begin{minipage}[c]{0.45\textwidth}
\centerline{\includegraphics[width = 8.1cm]{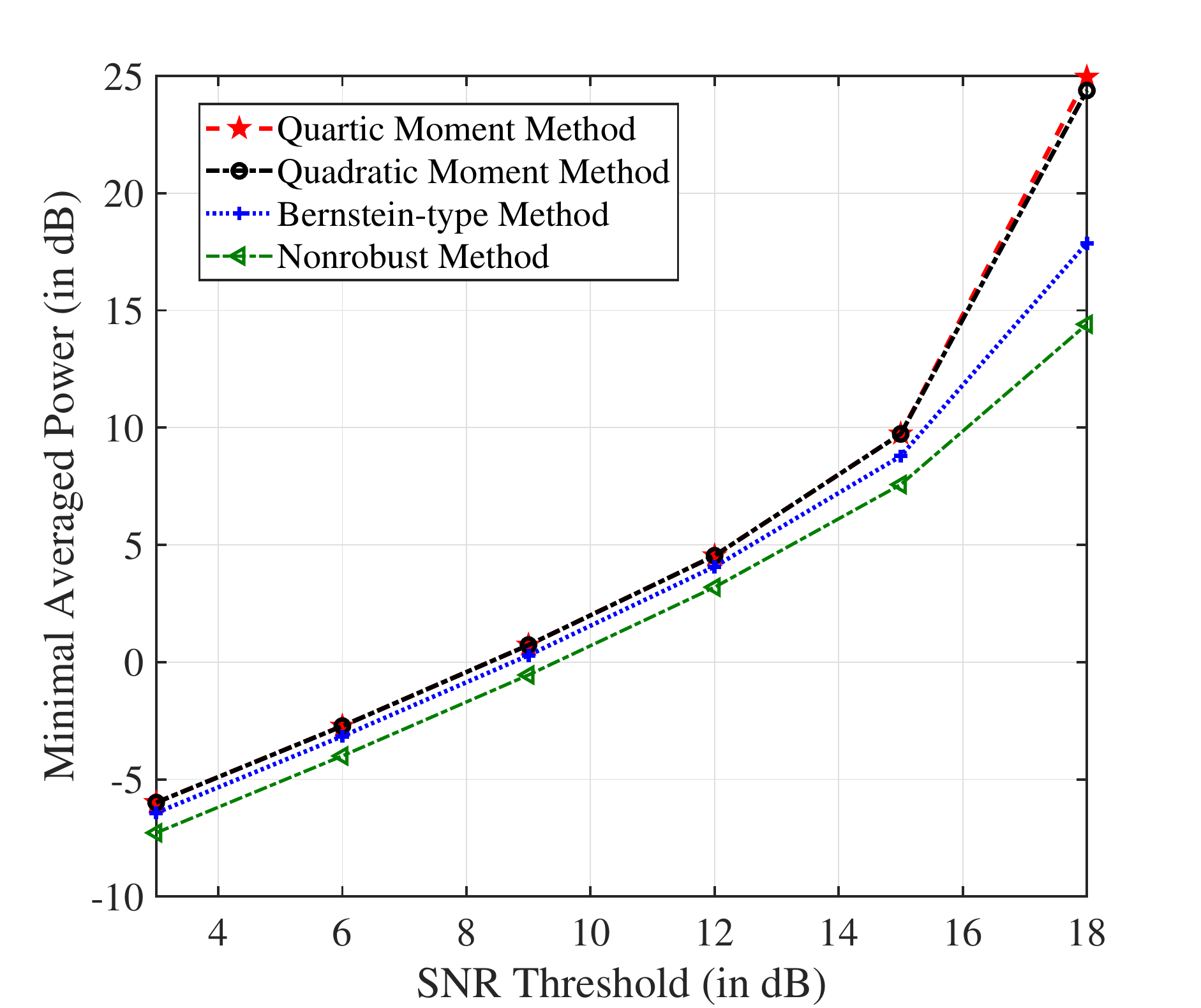}}
\caption{Minimum Power required versus the SNR threshold. For $\epsilon^2=\eta^2=0.002$; SNR outage percentage $\rho = 0.1$.}
\label{Fig:1}
\end{minipage}
\vspace{18pt}
\hfill
\begin{minipage}[c]{0.45\textwidth}
\centering
\centerline{\includegraphics[width = 8.1cm]{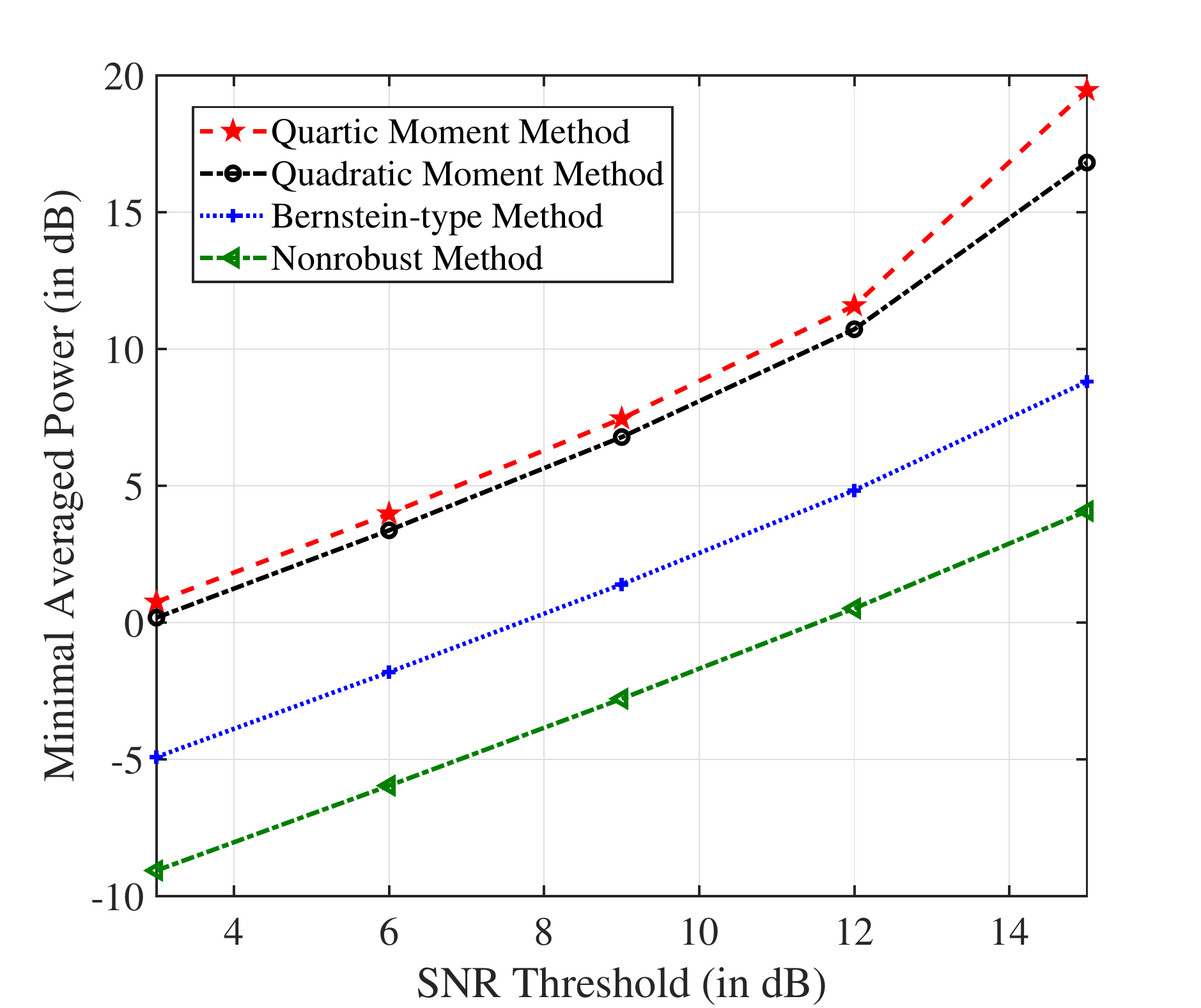}}
\caption{Minimum Power required versus the SNR threshold. For $\epsilon^2=\eta^2=0.06$; SNR outage percentage $\rho = 0.1$.}
\label{Fig:3}
\end{minipage}
\hfill
\begin{minipage}[c]{0.45\textwidth}
\centering
\centerline{\includegraphics[width = 8.4cm]{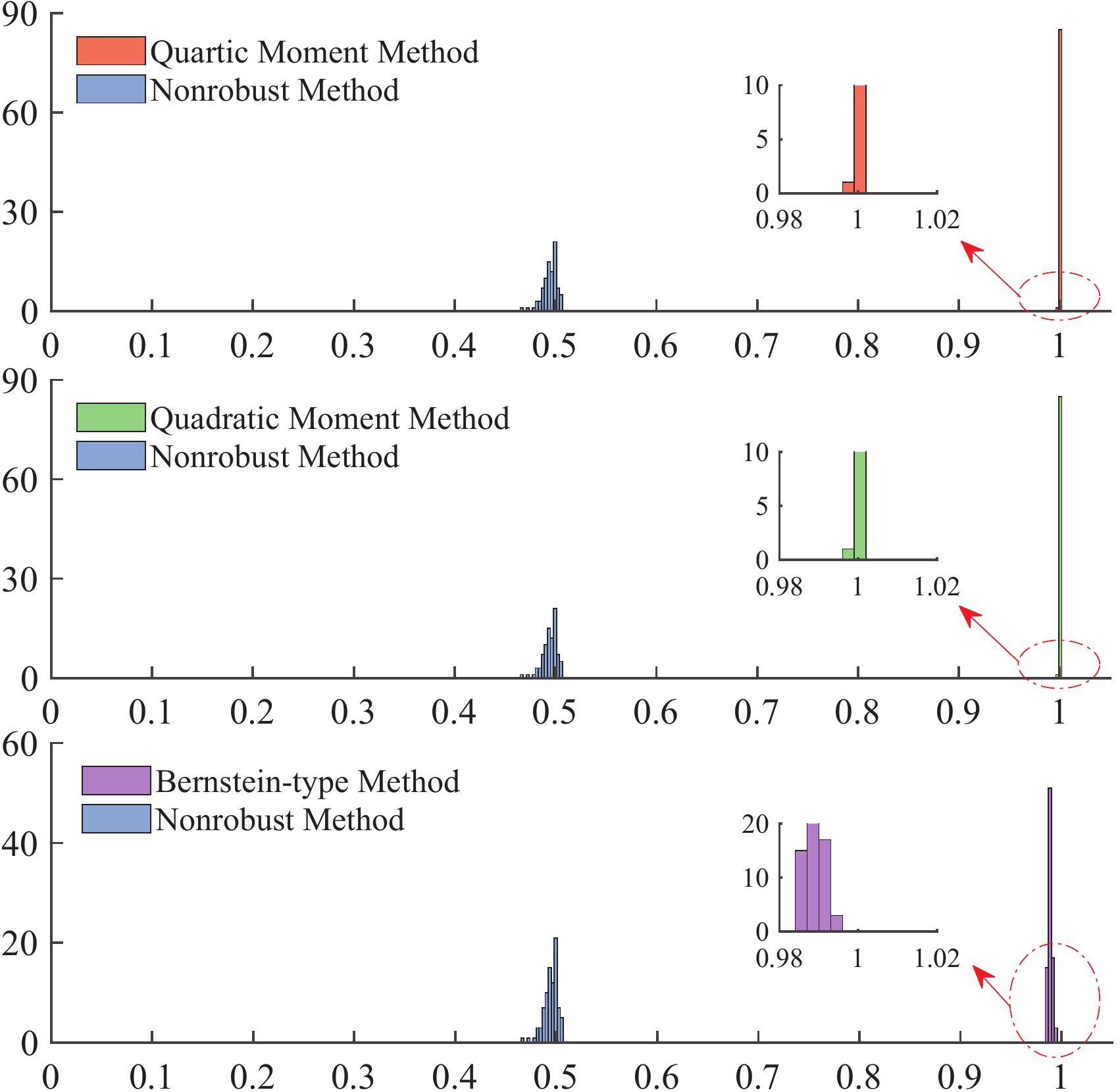}}
\caption{Histogram of the SNR satisfaction probability. For $\epsilon^2=\eta^2=0.002$; SNR outage percentage $\rho = 0.1$; SNR threshold $\gamma=18$dB.}
\label{Fig:2}
\end{minipage}
\hfill
\begin{minipage}[c]{0.45\textwidth}
\centering
\centerline{\includegraphics[width = 8.4cm]{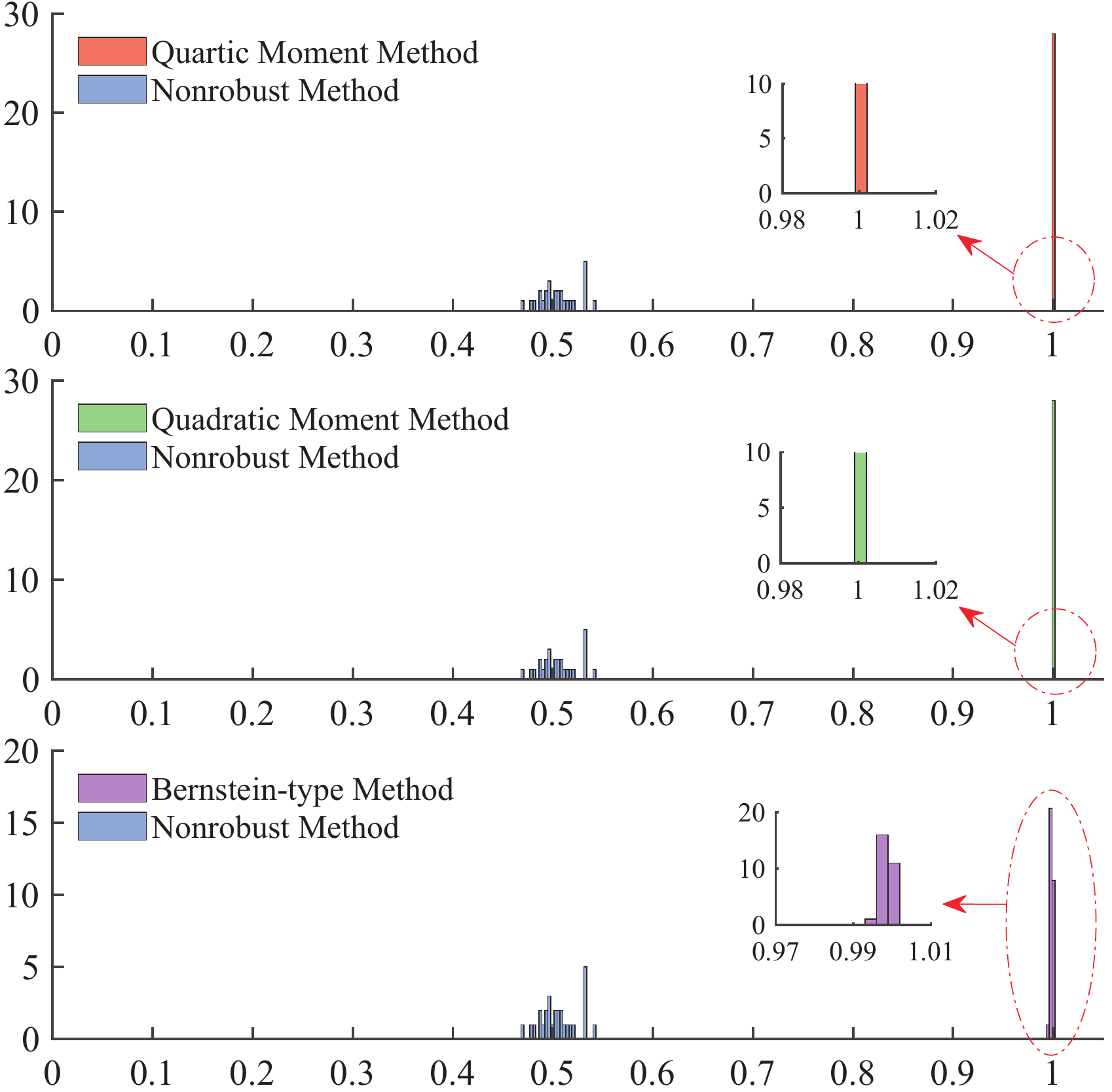}}
\caption{Histogram of the SNR satisfaction probability. For $\epsilon^2=\eta^2=0.06$; SNR outage percentage $\rho = 0.1$; SNR threshold $\gamma=12$dB.}
\label{Fig:4}
\end{minipage}
\end{figure*}

\section{Numerical Simulations}\label{sec:sim}
In this section, we provide numerical simulations to compare four different AF weight designs for our target problem, namely, the non-robust (NR) AF weight, the fourth-order moment inequality-based (M4) AF weight, the second-order moment inequality-based (M2) AF weight, and the Bernstein-type inequality-based (B2) AF weight. The setup of the experiments is as follows. The number of relays is set to be $L=4$; channels are generated by $\bar{\bm f} \sim \mathcal{CN}(0,{\bm I})$, $\bar{\bm g} \sim \mathcal{CN}(0,{\bm I})$ independently; channel errors are generated by $\Delta {\bm f} \sim \mathcal{CN}(0,\epsilon^2{\bm I})$, $\Delta {\bm g} \sim \mathcal{CN}(0,\eta^2{\bm I})$ independently, where the variances are specified in the sequel; the noise power at each relay is set to be $\sigma_\ell^2=0.25, \forall \ell$; the noise power at the receiver is $\sigma_{v}^2=0.25$; the outage probability is denoted by $\rho$. In this paper, we take $\rho=0.1$ to illustrate the performance of each scheme. Also, if the $kth$ largest eigenvalue of the SDR solution is $10^4$ times larger than the $(k+1)st$ largest eigenvalue, then the SDR solution is considered to be of rank-$k$.

\subsection{Comparison among the M4, M2, and B2 Approaches}
In Fig. \ref{Fig:1}, we present the averaged minimum power budget needed to satisfy the outage constraint as the SNR threshold varies from 3dB to $18$dB when $\rho=0.1$, $\epsilon^2=\eta^2=0.002$. Specifically, we compare the results for the non-robust and robust cases under different design approaches. In Fig. \ref{Fig:1}, we may consider that the NR design discards all perturbation terms, and thus the required power budget serves as a lower bound for all the robust designs.
We find that the M4 and M2 approaches require similar minimum power budgets, while the B2 and NR approaches require slightly less power budget. This is because the channel errors here are rather small and all the robust designs can easily handle the uncertainties in this case. This implies that the higher-order perturbation terms do not have much effect on the SNRs if the error is small.
Despite the similar power budgets, we can still see the differences of the tightness among the four design approaches. In Fig. \ref{Fig:2}, we test the outage cases of each approach by generating $1000$ complex Gaussian channel perturbations with specific variance and evaluating the SNR satisfaction probability for each channel realization (herein we calculate the SNR by \eqref{SNR_exact}). For each approach, we generate $100$ channel realizations and then pick up those realizations that are feasible for all methods to get the corresponding histograms in Fig. \ref{Fig:2}.  As shown in Fig. \ref{Fig:2},  the SNR satisfaction percentage of NR is the lowest, basically near $50\%$, which implies that the NR design is not reliable in general. On the contrary, all robust designs are very conservative and almost always satisfy the SNR constraint. It reveals that both the moment inequality-based and Bernstein-type inequality-based schemes are robust against channel errors. However, it is important to note from the histogram that compared to the M4 and M2 designs, the B2 design sometimes violates the SNR constraint. This implies that the moment inequality-based approaches are more conservative than the Bernstein-type inequality-based approach. This is consistent with our relative tightness analysis in Section \ref{sec:5}, although the outage rate $\rho$ is not within the target region.

In Fig. \ref{Fig:3} and Fig. \ref{Fig:4}, we increase the variances of the channel errors to be $0.06$. Since in this case the channel uncertainty is significantly enlarged, we can see an obvious power budget gap among different strategies. We find that M4 requires the largest power to support the robust design since the fourth-order moment inequality is the strictest. B2 requires a relatively lower power budget than M2 and M4, which implies that moment inequality-based approaches are more conservative than the Bernstein-type inequality-based approach. This is confirmed by the SNR satisfaction histogram (herein we calculate the SNR by \eqref{SNR_exact}). The results in Fig. \ref{Fig:4} again demonstrate the robustness of the moment inequality-based and Bernstein-type inequality-based approaches, as the SNR satisfaction percentages of B2, M2, and M4 all exceed our target threshold $0.9$. Actually, in this setting, the M2 and M4 designs always satisfy the SNR, while the B2 design sometimes violates the SNR constraint. This also confirms that the B2 approach is less conservative than the M2 and M4 approaches.

To further investigate conservatism, we increase the variance of channel uncertainty and compare the feasibility percentage, rank-one percentage, and randomization feasibility percentage in Tables \ref{tab:feasibility06} and \ref{tab:feasibility08}. In Table \ref{tab:feasibility06}, we have $\rho=0.1$ and $\epsilon^2=\eta^2=0.06$; in Table \ref{tab:feasibility08}, we have 
have $\rho=0.1$ and $\epsilon^2=\eta^2=0.08$. Herein, ``Feasibility\%" denotes the feasibility of the design problems \eqref{main2}, \eqref{main2_quad} and \eqref{main3}; ``Rank $k\%$" denotes the percentage of rank-$k$ solutions of the original design problem; ``Feasibility of Rand. in Rank $k$ Cases $\%$" denotes the rate of feasibility of the randomization algorithm for rank-$k$ solutions, where we set the number of randomizations to be $1000$. Note that ``NaN"  in the table indicates that the percentage of a certain rank is $0$, and thus there is no randomization of cases for this rank. The data in both tables reveal the same fact, i.e., M2 is less conservative than M4, while B2 is less conservative than M2 and M4.  This is consistent with our analytical results in Section \ref{sec:5},  although the outage rate $\rho$ for these cases is not within the target region in Theorem \ref{thm:2}.

\subsection{Comparison between M4 and M2 Approaches under Mismatched Noise and Perturbations}
To demonstrate that the M4 design is more robust than the M2 one, we set up three types of experiments. The first one follows the same setting as that in Fig. \ref{Fig:3}, where we have $\sigma_\ell^2 = \sigma_v^2=0.25$, $\epsilon^2 = \eta^2 = 0.002$, $\rho = 0.1$, and $\gamma =18$dB. We generate $1000$ and $10000$ channel realizations, respectively. For each channel realization, we solve M4 and M2 respectively to obtain the AF weights.  
	For the 1000 (resp. 10000) channel realizations generated, 861 (resp. 8725) of them are feasible for both the original design problems \eqref{main2} and \eqref{main2_quad}, and their optimal solutions are rank-one. Then, for each feasible channel realization, we generate $10000$ perturbation realizations to calculate the actual SNR (we calculate the SNR by \eqref{SNR_exact}) satisfaction rate $\hat\rho$ under M4 and M2 AF weights, where $\hat\rho$ is calculated as the ratio of the SNR-satisfying cases and the number of perturbation realizations (10000). In Table \ref{tab:different_channel}, we show the number of channel realizations that fall in different intervals of $\hat\rho$. Clearly, both M4 and M2 approaches provide a good outage rate, which is much less than $0.1$. However, we can still see that M4 is more restricted than M2, as the former has more channel realizations that exhibit higher $\hat\rho$ under the given perturbation realizations.

In the second experiment, we consider the scenario that the actual noise levels $\hat\sigma_{{\ell}}^2$ and $\hat\sigma_{{v}}^2$, or the actual perturbations $\hat\epsilon$ and $\hat\eta$, do not match the prior information in the original design problem \eqref{main0}. That is, when we solve \eqref{main0}, we set $\sigma_\ell^2 = \sigma_v^2=0.25$, $\epsilon^2 = \eta^2 = 0.002$, $\rho = 0.1$, and $\gamma =18$dB, while in practice, the noise levels change to  $\hat\sigma_\ell^2$ and $\hat\sigma_v^2$; or the perturbation levels change to $\hat\epsilon^2$ and $\hat \eta^2$.  Under the interference of the mismatch, SNR satisfaction rates and the corresponding number of channel realizations are shown in Table \ref{tab:different_perturbation}. We find that under M4, there are more channel realizations that exhibit higher $\hat\rho$ under the given mismatched realizations, which means that M4 is more robust than M2.

To further corroborate the robustness of M4, in Table \ref{tab:critical_point}, we investigate the critical point at which the mismatch of the noise or the perturbation causes the outage. To proceed,  we set $\rho = 0.1, \gamma =18$dB and solve the M4 and M2 design problems. Among $1000$ channel realizations, there are in total $861$ feasible cases for both M4 and M2. From the table, we can clearly see that as the mismatch increases, when $\hat\sigma_\ell^2=\hat\sigma_v^2=0.265$ or $\hat\epsilon^2=\hat \eta^2=0.0032$, M2 may cause outage in one out of $861$ realizations while M4 does not cause outage in all $861$ realizations. This implies that M4 is more robust than M2, as the former takes into account the exact SNR expression by keeping all higher-order perturbations.

\ifconfver
\begin{table*}[t!]
	\else
	\begin{table}[t!]
		\fi
		\centering
		\caption{Feasibility rate, rank-$k$ rate and feasibility rate of Randomization (Rand.):  $\sigma_\ell^2=0.25, \forall \ell$, $\sigma_v^2=0.25$,
			$\epsilon^2 = \eta^2 = 0.06$, $\rho = 0.1$.}
		\label{tab:feasibility06}
		\linespread{1.25} \rm \footnotesize 
		\setlength{\tabcolsep}{1.2mm}{
			\begin{tabular}{|c|c|c|c|c|c|c|c|c|c|c|c|c|c|c|c|}
				\hline
				\textbf{SNR in dB}                                                                                                          & \multicolumn{3}{c|}{$\gamma=3$}                  & \multicolumn{3}{c|}{$\gamma=6$}                  & \multicolumn{3}{c|}{$\gamma=9$}                  & \multicolumn{3}{c|}{$\gamma=12$}                 & \multicolumn{3}{c|}{$\gamma=15$}                 \\ \hline
				\textbf{Method}                                                                                                             & \textbf{M4} & \textbf{M2} & \textbf{B2} & \textbf{M4} & \textbf{M2} & \textbf{B2} & \textbf{M4} & \textbf{M2} & \textbf{B2} & \textbf{M4} & \textbf{M2} & \textbf{B2} & \textbf{M4} & \textbf{M2} & \textbf{B2} \\ \hline
				\textbf{Feasibility\%}                                                                                                      & $0.45$        & $0.50$         & $0.96$        & $0.42$        & $0.47$        & $0.92$        & $0.36$        & $0.40$         & $0.82$         & $0.28$        & $0.29$        & $0.78$        & $0.22$        & $0.24$        & $0.63$        \\ \hline
				\textbf{Rank $1$\%}                                                                                                         & $1.000$           & $1.000$           & $0.844$       & $1.000$           & $1.000$           & $0.891$       & $0.972$       & $1.000$           & $0.951$        & $1.000$           & $1.000$           & $0.961$       & $1.000$           & $1.000$           & $0.984$           \\ \hline
				\textbf{Rank $2$\%}                                                                                                         & $0$           & $0$           & $0.031$       & $0$           & $0$           & $0.022$       & $0.028$       & $0$           & $0.012$        & $0$           & $0$           & $0.026$       & $0$          & $0$           & $0.016$           \\ \hline
				\textbf{Rank $3$\%}                                                                                                       & $0$           & $0$           & $0.010$       & $0$           & $0$           & $0.065$       & $0$           & $0$           & $0.037$           & $0$           & $0$           & $0.013$       & $0$           & $0$           & $0$           \\ \hline
				\textbf{Rank $4$\%}                                                                                                        & $0$           & $0$           & $0.115$       & $0$           & $0$           & $0.022$       & $0$           & $0$           & $0$        & $0$           & $0$           & $0$           & $0$           & $0$           & $0$           \\ \hline
				\multicolumn{1}{|c|}{\textbf{\begin{tabular}[c]{@{}c@{}} Feasibility of Rand. \\ in Rank $2$ Cases\%\end{tabular}}}   & NaN         & NaN         & $0.0507$       & NaN         & NaN         & $0.0280$      & $1.0000$           & NaN         & $0.0750$      & NaN         & NaN         & $0.0345$      & NaN         & NaN         & $0.004$         \\ \hline
				\multicolumn{1}{|c|}{\textbf{\begin{tabular}[c]{@{}c@{}}Feasibility of Rand. \\ in Rank $3$ Cases\%\end{tabular}}} & NaN         & NaN         & $0.0020$      & NaN         & NaN         & $0.0005$           & NaN         & NaN         & $0$         & NaN         & NaN         & $0$      & NaN         & NaN         & NaN         \\ \hline
				\multicolumn{1}{|c|}{\textbf{\begin{tabular}[c]{@{}c@{}} Feasibility of Rand.\\ in Rank $4$ Cases\%\end{tabular}}}   & NaN         & NaN         & $9 \times 10^{-5}$      & NaN         & NaN         & $0$      & NaN         & NaN         &  NaN           & NaN         & NaN         & NaN         & NaN         & NaN         & NaN         \\ \hline
		\end{tabular}}
		\ifconfver
	\end{table*}
	\else
\end{table}
\fi

\ifconfver
\begin{table*}[t!]
	\else
	\begin{table}[t!]
		\fi
		\centering		
		\caption{Feasibility rate, rank-$k$ rate and feasibility rate of Randomization (Rand.): $\sigma_\ell^2=0.25, \forall \ell$, $\sigma_v^2=0.25$,
			$\epsilon^2 = \eta^2 = 0.08$, $\rho = 0.1$.}
		\label{tab:feasibility08}
		\linespread{1.25} \rm \footnotesize 
		\setlength{\tabcolsep}{1.2mm}{
			\begin{tabular}{|c|c|c|c|c|c|c|c|c|c|c|c|c|c|c|c|}
				\hline
				\textbf{SNR in dB}                                                                                                          & \multicolumn{3}{c|}{$\gamma=3$}                  & \multicolumn{3}{c|}{$\gamma=6$}                  & \multicolumn{3}{c|}{$\gamma=9$}                  & \multicolumn{3}{c|}{$\gamma=12$}                 & \multicolumn{3}{c|}{$\gamma=15$}                 \\ \hline
				\textbf{Method}                                                                                                             & \textbf{M4} & \textbf{M2} & \textbf{B2} & \textbf{M4} & \textbf{M2} & \textbf{B2} & \textbf{M4} & \textbf{M2} & \textbf{B2} & \textbf{M4} & \textbf{M2} & \textbf{B2} & \textbf{M4} & \textbf{M2} & \textbf{B2} \\ \hline
				\textbf{Feasibility\%}                                                                                                      & $0.17$        & $0.19$        & $0.91$        & $0.16$        & $0.19$        & $0.81$        & $0.16$        & $0.16$        & $0.72$        & $0.10$         & $0.13$        & $0.70$        & $0.06$        & $0.06$        & $0.44$         \\ \hline
				\textbf{Rank $1$\%}                                                                                                         & $1.000$           & $1.000$           & $0.714$       & $1.000$           & $1.000$           & $0.840$       & $1.00$           & $1.00$           & $0.833$       & $1.000$           & $1.000$           & $0.871$       & $1.000$           & $1.000$           & $0.864$         \\ \hline
				\textbf{Rank $2$\%}                                                                                                         & $0$           & $0$           & $0.077$       & $0$           & $0$           & $0.074$           & $0$           & $0$           & $0.115$           & $0$           & $0$           & $0.086$       & $0$           & $0$           & $0.114$         \\ \hline
				\textbf{Rank $3$\%}                                                                                                       & $0$           & $0$           & $0.033$       & $0$           & $0$           & $0.025$       & $0$           & $0$           & $0.039$       & $0$           & $0$           & $0.029$       & $0$           & $0$           & $0$         \\ \hline
				\textbf{Rank $4$\%}                                                                                                        & $0$           & $0$           & $0.176$       & $0$           & $0$           & $0.061$       & $0$           & $0$           & $0.013$       & $0$           & $0$           & $0.014$           & $0$           & $0$           & $0.022$         \\ \hline
				\multicolumn{1}{|c|}{\textbf{\begin{tabular}[c]{@{}c@{}}Feasibility of Rand. \\ in Rank $2$ Cases\%\end{tabular}}}   & NaN         & NaN         & $0.0733$       & NaN         & NaN         & $0.0583$         & NaN         & NaN         & $0.0453$         & NaN         & NaN         & $0.0187$           & NaN         & NaN         & $0.004$        \\ \hline
				\multicolumn{1}{|c|}{\textbf{\begin{tabular}[c]{@{}c@{}}Feasibility of Rand. \\ in Rank $3$ Cases\%\end{tabular}}} & NaN         & NaN         & $0.0097$       & NaN         & NaN         & $0.0020$      & NaN         & NaN         & $0.0033$      & NaN         & NaN         & $0.0010$       & NaN         & NaN         & NaN      \\ \hline
				\multicolumn{1}{|c|}{\textbf{\begin{tabular}[c]{@{}c@{}}Feasibility of Rand.  \\ in Rank $4$ Cases\%\end{tabular}}}   & NaN         & NaN         & $0.0003$    & NaN         & NaN         & $0.0008$  & NaN         & NaN         & $0.0020$           & NaN         & NaN         & $0$         & NaN         & NaN         & $0.2320$           \\ \hline
		\end{tabular}}
		\ifconfver
	\end{table*}
	\else
\end{table}
\fi

\ifconfver
\begin{table*}[t!]
	\else
	\begin{table}[t!]
		\fi
		\centering	
		\caption{SNR Satisfaction rates for different channel realizations: $\sigma_\ell^2=0.25, \forall \ell$, $\sigma_v^2=0.25$,
			$\epsilon^2 = \eta^2 = 0.002$, $\rho = 0.1$; $\gamma=18$dB.}
		\label{tab:different_channel}
		\linespread{1.25} \rm \footnotesize 
		\setlength{\tabcolsep}{3.5mm}{
			\begin{tabular}{|c|c|c|c|c|}
				\hline
				\textbf{Number of Channel Realizations}                 & \multicolumn{2}{c|}{$1000$} & \multicolumn{2}{c|}{$10000$}  \\ \hline
				\diagbox[height=4em,width=16em,trim=lr]   {\textbf{ SNR Satisfaction Rate} ($\hat{\rho}$)} {\textbf{Method}}  & \textbf{M4}          & \textbf{M2}          & \textbf{M4}           & \textbf{M2}                   \\ \hline
				$1.000 \ge \hat{\rho} > 0.999 $                               & $852$         & $850$         & $8616$          & $8597$            \\ \hline
				$0.999 \ge \hat{\rho} > 0.998 $                                  & $7$           & $8$           & $86$          & $95$               \\ \hline
				$0.998 \ge \hat{\rho} > 0.997 $                                  & $2$           & $2$           & $8$ & $7$                  \\ \hline
				$0.997 \ge \hat{\rho} > 0.996 $                                   & $0$           & $0$           & $3$            & $7$                 \\ \hline
				$0.996 \ge \hat{\rho} > 0.995 $                                  & $0$           & $1$           & $1$            & $3$                  \\ \hline
				$0.995 \ge \hat{\rho} > 0.994 $                                   & $0$           & $0$           & $2$            & $3$                 \\ \hline
				$0.994 \ge \hat{\rho} > 0.993 $                                  & $0$           & $0$           & $1$            & $3$                  \\ \hline
				$0.993 \ge \hat{\rho} > 0.992 $                                  & $0$           & $0$           & $1$            & $2$                 \\ \hline
				$0.992 \ge \hat{\rho} > 0.991 $                                & $0$           & $0$           & $0$            & $1$             \\ \hline
				$0.991 \ge \hat{\rho} > 0.990$                                   & $0$           & $0$           & $1$            & $0$         \\ \hline
				$0.990 \ge \hat{\rho} \ge 0$                                      & $0$          & $0$           & $6$            & $7$               \\ \hline
				\textbf{Number of Feasible Channel Realizations}   & \multicolumn{2}{c|}{$861$} & \multicolumn{2}{c|}{$8725$}  \\ \hline
		\end{tabular}}
		\ifconfver
	\end{table*}
	\else
\end{table}
\fi

\ifconfver
\begin{table*}[t!]
	\else
	\begin{table}[t!]
		\fi
		\centering	
		\caption{SNR Satisfaction rates under noise and perturbation mismatches: $\rho = 0.1$, $\gamma=18$dB, channel realization $= 1000$.}
		\label{tab:different_perturbation}
		\linespread{1.25} \rm \footnotesize 
		\setlength{\tabcolsep}{1.5mm}{
			\begin{tabular}{|c|c|c|c|c|c|c|c|c|c|c|c|c|c|c|}
				\hline
				
				\textbf{~~} & \multicolumn{2}{c|}{\textbf{\begin{tabular}[c]{@{}c@{}}Original  \\ (No Mismatch)\end{tabular}}} & \multicolumn{6}{c|}{\textbf{Noise Mismatch}}                                                  & \multicolumn{6}{c|}{\textbf{ Perturbation Mismatch}}                                               \\ \hline
				\textbf{Channel Error ($\hat\eta^{2} =\hat\epsilon^{2} $)}     & \multicolumn{8}{c|}{0.002}                                                                                           & \multicolumn{2}{c|}{0.0022} & \multicolumn{2}{c|}{0.0024} & \multicolumn{2}{c|}{0.0028} \\ \hline
				\textbf{Noise} ($\hat\sigma_\ell^2 = \hat\sigma_v^2 $)             & \multicolumn{2}{c|}{0.25}     & \multicolumn{2}{c|}{0.252} & \multicolumn{2}{c|}{0.256} & \multicolumn{2}{c|}{0.258} & \multicolumn{6}{c|}{0.25}                                                               \\ \hline
				\diagbox[height=4em,width=16em,trim=lr]  { \textbf{SNR Satisfaction Rate ($\hat{\rho}$)}}{\textbf{Method}}           & \textbf{M4}             & \textbf{M2}    & \textbf{M4}                       & \textbf{M2}          & \textbf{M4}            & \textbf{M2}         & \textbf{M4}            & \textbf{M2}           & \textbf{M4}            & \textbf{M2}            & \textbf{M4}            & \textbf{M2}          & \textbf{M4}            & \textbf{M2}           \\ \hline
				$1.000 \ge \hat{\rho} > 0.999 $                  & $852$           & $850$           & $703$          & $694$         & $65$           & $59$          & $31$           & $26$          & $306$          & $294$          & $1$            & $0$            & $1$            & $0$            \\ \hline
				$0.999 \ge \hat{\rho} > 0.998 $             & $7$             & $8$             & $156$          & $165$         & $285$          & $283 $        & $85$           & $80$          & $538$          & $550$          & $130$          & $122$          & $0$            & $1$            \\ \hline
				$0.998 \ge \hat{\rho} > 0.997 $              & $2$             & $2$             & $1$            & $1$           & $343$          & $346$         & $163$          & $158$         & $12 $          & $12$           & $508$          & $503$          & $0$           & $0$            \\ \hline
				$0.997 \ge \hat{\rho} > 0.996 $             & $0$             & $0$             & $0$            & $0$           & $127$          & $131$         & $203$          & $206$         & $2 $           & $1$            & $206$          & $215$          & $2$            & $1$            \\ \hline
				$0.996 \ge \hat{\rho} > 0.995 $              & $0$             & $1$             & $1$            & $0$           & $33$           & $33$          & $182$          & $181$         & $1$            & $1$            & $11$           & $13$           & $7$            & $6$            \\ \hline
				$0.995 \ge \hat{\rho} > 0.994 $              & $0$             & $0$             & $0$            & $1$           & $4$            & $5$           & $111$          & $119$         & $1$            &$ 1$            & $0$            & $2$            & $19$           & $14$           \\ \hline
				$0.994 \ge \hat{\rho} > 0.993 $  $$            &$ 0$             & $0$             & $0$            & $0$           & $2$            & $$2$$           & $36$          & $37$         & $0$            & $1$            & $1$            & $0$            & $74$           & $69$           \\ \hline
				$0.993 \ge \hat{\rho} > 0.992 $              & $0 $            & $0 $            &$ 0$            & $0$           & $1$            & $1$          & $25$           & $27 $         & $0$            &$ 0 $           & $1 $           &$ 1$            & $182$          & $183$          \\ \hline
				$0.992 \ge \hat{\rho} > 0.991 $              & $0 $            & $0$             & $0$            & $0$           & $1$            & $1$           & $15$           & $17$          & $0$            & $0$            & $0$            & $2$            & $293$          & $292$          \\ \hline
				$0.991 \ge \hat{\rho} > 0.990 $              & $0 $            & $0$             & $0$            & $0$           & $0$            & $0$           & $4$            & $4$          & $0$            & $0$            & $0$            & $0$            & $195$          & $200$          \\ \hline
				$0.990 \ge \hat{\rho} \ge 0 $              & $0$             & $0$             & $0$            & $0$           & $0$            & $0$           & $6$            & $6$           & $1$            & $1$            & $3$            & $3$            & $88$           & $95$           \\ \hline
				\textbf{Number of Feasible Channel Realizations} & \multicolumn{14}{c|}{$861$}  \\ \hline
				\end{tabular}}
			\ifconfver
			\end{table*}
			\else
			\end{table}
			\fi

			\ifconfver
			\begin{table*}[t!]
				\else
				\begin{table}[t!]
					\fi
					\centering	
					\caption{Noise mismatch and perturbation mismatch cause the outage: $\rho = 0.1$,  $\gamma=18$dB, channel realization $= 1000$.}
					\label{tab:critical_point}
					\linespread{1.25} \rm \footnotesize 
					\setlength{\tabcolsep}{3.5mm}{
						\begin{tabular}{|c|c|c|c|c|c|c|}
							\hline
							\textbf{~} & \multicolumn{2}{c|}{\tabincell{c}{\textbf{Original}\\ \textbf{(No Mismatch)}}} & \multicolumn{2}{c|}{\textbf{Noise Mismatch}} & \multicolumn{2}{c|}{\textbf{Perturbation Mismatch}} \\ \hline
							\textbf{Noise ($\hat\sigma_\ell^2 = \hat\sigma_v^2$)}              & \multicolumn{2}{c|}{0.25}   & \multicolumn{2}{c|}{0.265}          & \multicolumn{2}{c|}{0.25}                 \\ \hline
							\textbf{Channel Error ($\hat\eta^{2} =\hat\epsilon^{2} $)}      & \multicolumn{4}{c|}{0.002}                                        & \multicolumn{2}{c|}{0.0032}               \\ \hline
							\diagbox[height=4em,width=16em,trim=rl]  {  \textbf{SNR Satisfaction Rate ($\hat{\rho}$)}} { \textbf{Method}}            & M4           & M2           & M4               & M2               & M4                  & M2                  \\ \hline
							$1.000 \ge \hat{\rho} \ge 0.900 $                  & $861$          & $861$          & $861 $             & $860$              & $861$                 & $860$                 \\ \hline
							$0.900 > \hat{\rho} \ge 0 $ (Outage)                  & $0$            & $0$            & $0$                & $1$                & $0$                   & $1$                   \\ \hline
							\textbf{Number of Feasible Channel Realizations} & \multicolumn{6}{c|}{$861$}  \\ \hline
					\end{tabular}}
					\ifconfver
				\end{table*}
				\else
			\end{table}
			\fi

\subsection{Relative Tightness Verification for M2 and B2 Approaches}
Lastly, we set the outage probability $\rho=0.00044$ to verify the conclusion in  Theorem  \ref{thm:2}, i.e., ``\emph{for any $\rho \in (\exp(-8),0.00045)$, the moment inequality-based safe approximation is always more conservative than the Bernstein-type inequality-based safe approximation}."
We also set $\sigma_\ell^2 = \sigma_v^2=0.25$ and $\epsilon^2=\eta^2=0.0005$. The results are presented in Fig. \ref{Fig:5}, Fig. \ref{Fig:6}, and Table \ref{tab:feasibility00044}. We find that M2 requires more power than B2, which implies that moment inequality-based approaches are more conservative than the Bernstein-type inequality-based approach. The SNR satisfaction percentages of both B2 and M2 exceed our target threshold $1-0.00044$ (herein we calculate the SNR by \eqref{SNR_B2}), as M2 always provides  $100\%$ SNR satisfaction while B2 provides a slightly lower SNR satisfaction percentage. It is interesting to see that B2 requires a similar level of power as the non-robust case, while its SNR satisfaction is significantly better than the latter. This shows the necessity of the robust design. Table \ref{tab:feasibility00044} is even more interesting, as the B2 design seems always feasible while the M2 design is rarely feasible. Note that the rank-one feasibility is $100\%$ for both B2 and M2. We can clearly see and conclude that the moment inequality-based approach is more conservative than the Bernstein-type inequality-based approach from this experiment, which is consistent with our Theorem \ref{thm:2}.

\begin{figure}[t!]
	\centering
	\centerline{\includegraphics[width = 8.5cm]{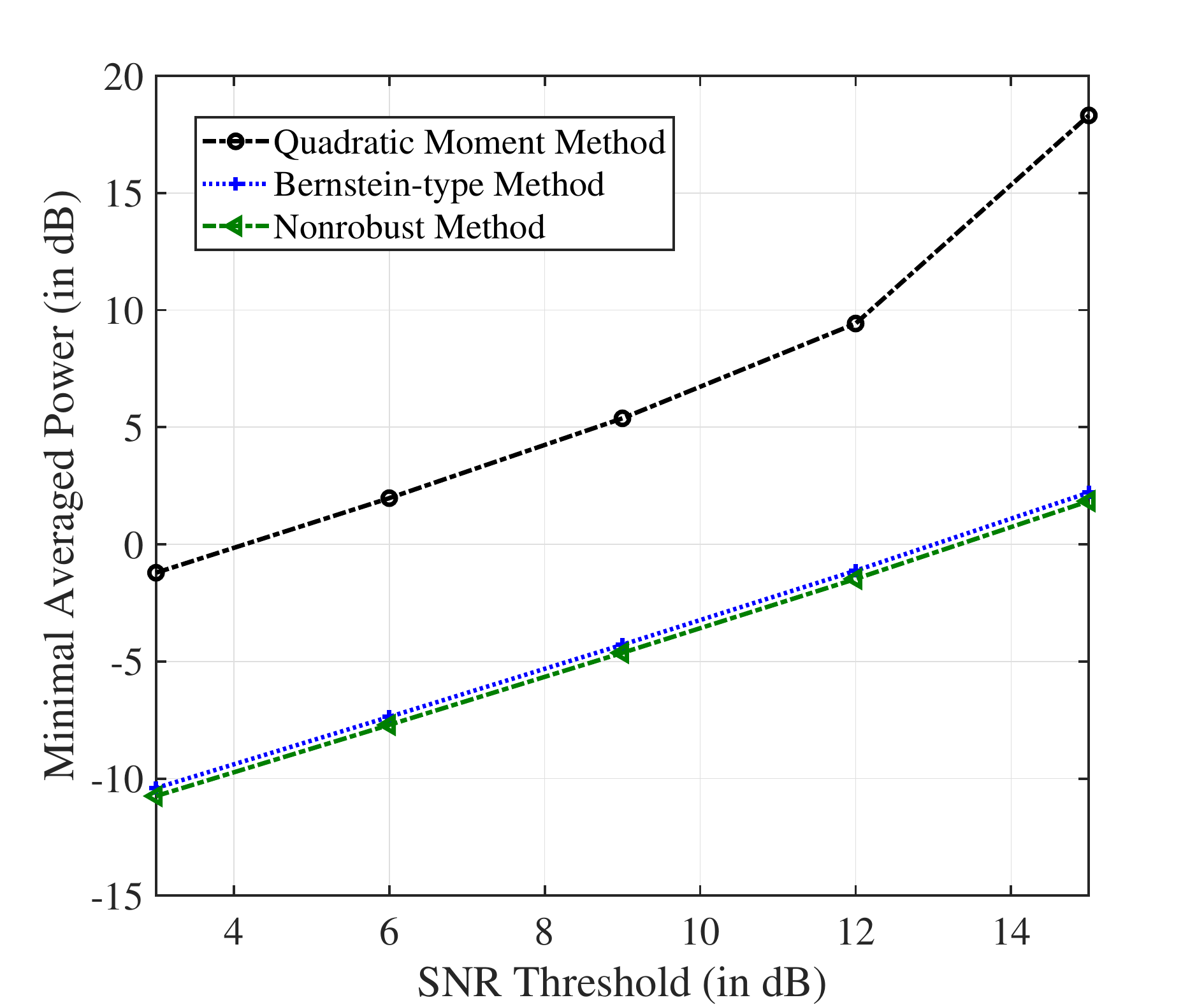}}
	\caption{Minimum Power required versus the SNR threshold. For $\epsilon^2=\eta^2=0.0005$; SNR outage percentage $\rho = 0.00044$.}
	\label{Fig:5}
\end{figure}
\begin{figure}[t!]
	\centering
	\centerline{\includegraphics[width = 8.2cm]{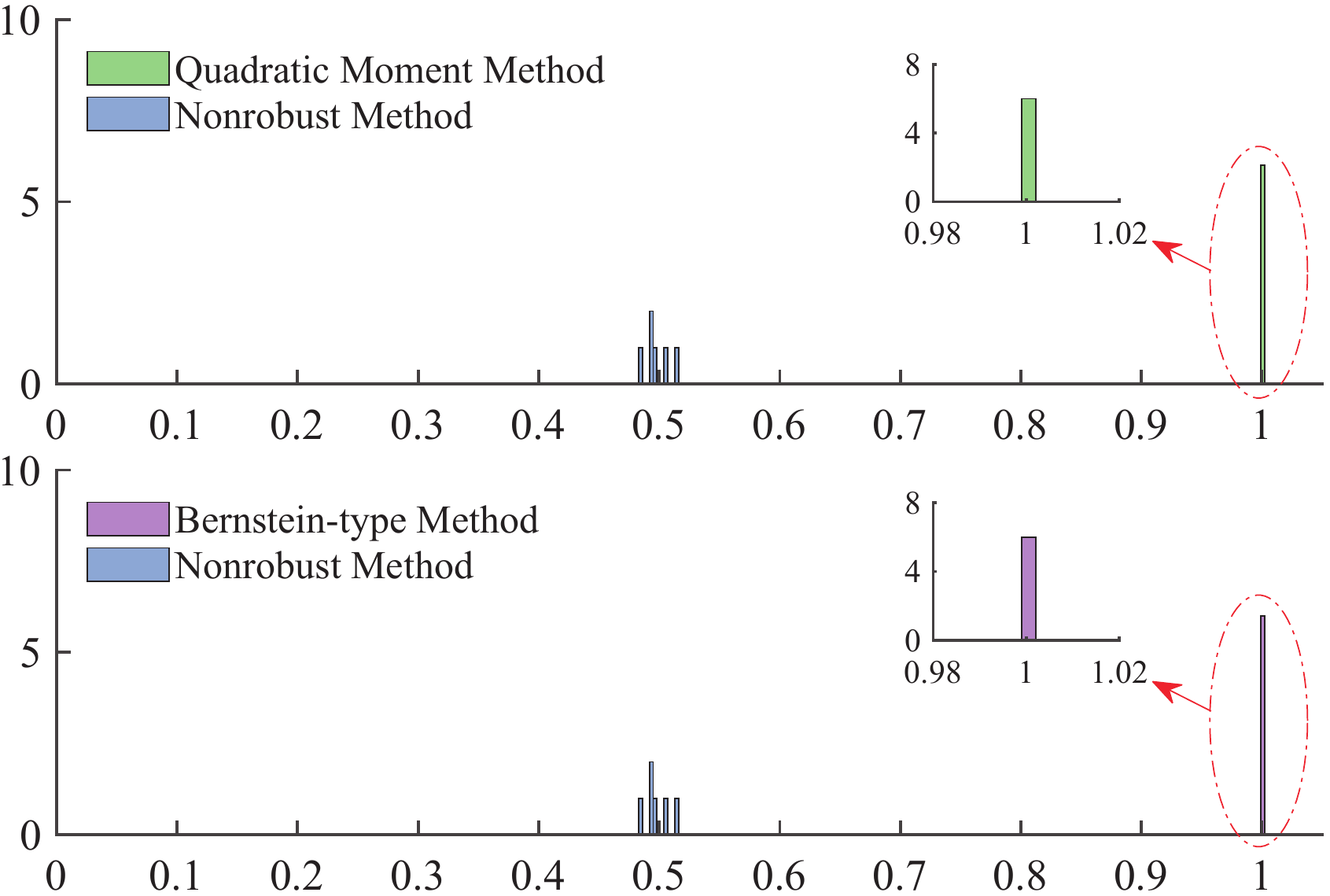}}
	\caption{Histogram of the SNR satisfaction probability. For $\epsilon^2=\eta^2=0.0005$; SNR outage percentage $\rho = 0.00044$; SNR threshold $\gamma=12$dB.}
	\label{Fig:6}
\end{figure}

\ifconfver
\begin{table*}[t!]
	\else
	\begin{table}[t!]
		\fi
		\centering		
		\caption{Feasibility rate and the rank-one rate respectively for $\sigma_\ell^2=0.25, \forall \ell$, $\sigma_v^2=0.25$,
			$\epsilon^2 = \eta^2 = 0.0005$; $\rho = 0.00044$.}
		\label{tab:feasibility00044}
		\linespread{1.25} \rm \footnotesize 
		\setlength{\tabcolsep}{4.5mm}{
			\begin{tabular}{|c|c|c|c|c|c|c|c|c|c|c|c|c|c|c|c|}
				\hline
				\textbf{SNR in dB}                                                                                                          & \multicolumn{2}{c|}{$\gamma=3$}                  & \multicolumn{2}{c|}{$\gamma=6$}                  & \multicolumn{2}{c|}{$\gamma=9$}                  & \multicolumn{2}{c|}{$\gamma=12$}                 & \multicolumn{2}{c|}{$\gamma=15$}                 \\ \hline
				\textbf{Method}                                                                                                             &\textbf{M2} & \textbf{B2}  & \textbf{M2} & \textbf{B2} &  \textbf{M2} & \textbf{B2} & \textbf{M2} & \textbf{B2}  & \textbf{M2} & \textbf{B2} \\ \hline
				\textbf{Feasibility\%}                                                                                                      & $0.07$        & $1.00$        & $0.06$        & $1.00$        & $0.06$        & $1.00$        & $0.06$        & $1.00$        & $0.04$        & $1.00$                  \\ \hline
				\textbf{Rank One\%}                                                                                                         & $1.00$           & $1.00$           & $1.00$       & $1.00$           & $1.00$           & $1.00$       & $1.00$           & $1.00$           & $1.00$       & $1.00$                   \\ \hline
		\end{tabular}}
		\ifconfver
	\end{table*}
	\else
\end{table}
\fi

\section{Conclusions} \label{sec:con}
In this paper we studied the robust design problem for two-hop one-way relay beamforming. Specifically, we considered the scenario where both the transmitter-to-relay and relay-to-receiver links are subject to errors. This scenario is difficult and seldom studied in the literature as it involves chance constraints with quartic perturbations. We provided different reformulations of the chance-constrained robust design problem and further analyzed the relative tightness of different reformulations.
Numerical results further confirmed the superiority of the proposed robust design. The quartic perturbation-based outage-constrained robust design
is indeed more conservative. Nevertheless, by taking into account the higher-order perturbations, the resulting design is more robust against mismatch of prior distributional information. The SNR satisfaction rate and rank feasibility tables verified the tightness results. In the future, many transceiver pairs could be considered as a non-trivial extension of this work.

\appendix

\section{Appendix}

\subsection{Proof of Theorem~\ref{thm:1}}
\label{appendix:1}
Given $\bar{f}({\bm x},{\bm \xi}) = f({\bm x},{\bm \xi}) + a_0(\bm x)$, by assumption, for each ${\bm x} \in \C^n$ the function ${\bm x} \mapsto \bar{f}({\bm x},{\bm \xi})$
is affine in ${\bm x} \in \C^n$. This implies that ${\bm x} \mapsto \bar{f}^2({\bm x},{\bm \xi})$ is a
non-negative homogeneous quadratic polynomial in ${\bm x} \in \C^n$. 
This establishes (a) in Theorem~\ref{thm:1}.

To prove Theorem~\ref{thm:1}(b), we need the following lemma.

\begin{lemma} \label{thm:mom-bd}
	(cf.~\cite[Theorem 5.10]{janson1997gaussian}) For all $q\ge2$,
	$$ \E{|\bar{f}({\bm x},{\bm \xi})|^q}^{1/q} \le (q-1)^2\E{|\bar{f}({\bm x},{\bm \xi})|^2}^{1/2}. $$
\end{lemma}

To prove Theorem~\ref{thm:1}(b), since
$ \bar{f}({\bm x},{\bm \xi})^{2}={\bm v}^T({\bm x}){\bm U}({\bm \xi}){\bm v}({\bm x}),$
it follows that
 \begin{align*}
 \E{|\bar{f}({\bm x},{\bm \xi})|^2} \ge &
 \E{\bar{f}({\bm x},{\bm \xi})^2}\\
 =& {\bm v}^T({\bm x}) \E{{\bm U}({\bm \xi})} {\bm v}({\bm x}) \\=& {\bm v}^T({\bm x}) {\bm U} {\bm v}({\bm x}),
  \end{align*}
where ${\bm U} = \E{{\bm U}({\bm \xi})}$ is a Hermitian positive semidefinite matrix, which can be computed explicitly, as each entry of ${\bm U}({\bm \xi})$ involves only the expectation of a certain product of standard Gaussian random variables. 

%
%
For ease of notation, we let $\bar{q} = q(\rho)$, where $q(\rho)$ is defined in \eqref{q(epsilon)}. By Lemma \ref{thm:mom-bd} and Markov's inequality, for any $\bar q\ge2$, we have
\begin{eqnarray*}
		\Pr(|\bar{f}({\bm x},{\bm \xi})| \ge t) &\le& \frac{\E{|\bar{f}({\bm x},{\bm \xi})|^{\bar q}}}{t^{\bar q}} \\
		\noalign{\medskip}
		&\le& \frac{({\bar q}-1)^{2{\bar q}} \cdot \E{|\bar{f}({\bm x},{\bm \xi})|^2}^{{\bar q}/2}}{t^{\bar q}} \\
		\noalign{\medskip}
		&\le& \left( \frac{({\bar q}-1)^2 \cdot \|{\bm U}^{1/2} {\bm v}({\bm x})\|_2}{t} \right)^{\bar q}.
\end{eqnarray*}
Thus, whenever
	$
	t \ge  c(\rho) \|{\bm U}^{1/2} {\bm v}({\bm x})\|_2
	$,
	where $c(\rho)$ is defined in \eqref{c(epsilon)},
	we have $\Pr( |\bar{f}({\bm x},{\bm \xi})| \ge t ) \le \rho$. Whenever the second-order cone constraint in \eqref{SOC constrain}
	holds, we have
	\begin{align*}
	\Pr( f({\bm x},{\bm \xi}) \ge 0 ) = &\Pr( \bar{f}({\bm x},{\bm \xi}) \ge a_0(\bm x) ) \\
	\le& \Pr( |\bar{f}({\bm x},{\bm \xi})| \ge a_0(\bm x) ) \\
	\le &\Pr(|\bar{f}({\bm x},{\bm \xi})| \ge t)\\
	\le &\rho.
	\end{align*}
	This implies that the (complex) second-order cone constraint (\ref{SOC constrain}) is a safe tractable approximation of (\ref{chance_Q}), as desired.

\bibliographystyle{IEEEtran}
\bibliography{chance_relay_ref,relay_poly_iccasp_2016_ref}
\end{document}